\documentclass[10pt,journal,letterpaper]{IEEEtran}
\usepackage[T1]{fontenc}
\usepackage{newtxtext,amsmath,newtxmath}
\usepackage{booktabs,array,cite,url,balance}
\usepackage{tikz,pgfplots}
\pgfplotsset{compat=1.18}
\definecolor{paperblack}{HTML}{333333}
\definecolor{paperblue}{HTML}{0072B2}
\definecolor{paperred}{HTML}{D55E00}
\definecolor{papergreen}{HTML}{009E73}
\pgfplotsset{paperaxis/.style={tick label style={font=\footnotesize},label style={font=\footnotesize},title style={font=\small},grid=major,grid style={black!12},axis line style={black!60},legend cell align=left,legend style={draw=none,fill=white},scaled ticks=false}}

\newcommand{\col}{\operatorname{col}}

\newcommand{\clip}[1]{\left[#1\right]_0^1}
\newtheorem{theorem}{Theorem}
\newtheorem{corollary}{Corollary}
\newenvironment{proof}{\begin{IEEEproof}}{\end{IEEEproof}}
\begin{document}
\title{From Echoes to Evidence: Cooperative Fusion for Integrated Sensing and Communication}
\author{Xingqin~Lin
}

\markboth{}{}
\maketitle
\begin{abstract}
Cooperative integrated sensing and communication (ISAC) can support shared situational awareness for unmanned aerial vehicles (UAVs), road vehicles, and industrial mobile robots such as automated guided vehicles (AGVs). Yet multipath can make nodes agree on a virtual target while useful weak echoes remain unexploited. This paper develops cooperative fusion that separates association confirmation, state refinement, and residual recovery. A cellular reference signal model accommodates arbitrary node counts, three-dimensional motion, monostatic and bistatic links, nuisance parameters, and correlated errors. We characterize report distinguishability under bounded bias and shared-reflector ambiguities that survive additional nodes. Minimax estimation yields directional weak-report refinement that preserves confirmed associations and limits bias-sensitive innovations; acquisition accounts for request and upload costs. Controlled experiments accompany UAV evaluations under a Third Generation Partnership Project (3GPP) stochastic channel model and a ray tracing environment representing San Francisco. In the latter environment, matched-budget fusion raises detection from observable site selection's 88.75\% to 92.46\% and reduces the false-output fraction from 11.06\% to 0.40\%. Against the strongest full-band fixed site, the 90th-percentile position error on commonly detected targets falls from 8.54 to 3.50 m. Association controls, local rejection ablations, and refinement studies distinguish the sources of these gains.
\end{abstract}
\begin{IEEEkeywords}
Integrated sensing and communication, cellular reference signals, cooperative sensing, multipath, biased estimation, measurement fusion.
\end{IEEEkeywords}
\section{Introduction}

Integrated sensing and communication (ISAC) is a key direction for sixth-generation (6G) networks and beyond, enabling cellular infrastructure to carry data and observe its environment~\cite{chen_5ga6g,lin_bridge6g}. Applications include unmanned aerial vehicle (UAV) operations, vehicle awareness, and automated guided vehicle (AGV) coordination. Known reference symbols support delay, Doppler, and angle measurements from complementary node views. Distributed radar theory and cooperative ISAC establish spatial-diversity benefits~\cite{fishler_diversity,godrich_localization,meng_scaling}, which depend on visibility, scattering strength, and velocity observability.

Multipath creates an evidence problem: reflected reports can agree precisely on a virtual target, while a weak direct echo rejected locally can resolve another node's uncertainty. Recovery candidates may still bias confirmed states. We link association confirmation, refinement at acceptable bias risk, and additional acquisition through propagation ambiguity, report distinguishability, correlated estimation risk, and acquisition cost.

Waveform and receiver studies supply the first foundation. Orthogonal
frequency-division multiplexing (OFDM) supports communication and radar
processing \cite{sturm_ofdm}. Positioning reference signals (PRS) are one
cellular realization, with sensing processing and estimation bounds established
by Wei \emph{et al.} \cite{wei_prs}. Recent studies characterize random-data
sidelobes, shape constellations, and optimize range--Doppler sidelobes
\cite{keskin_monostatic,du_constellation,li_sidelobe}; tensor methods exploit
channel and target structure \cite{zhang_tensor}. Bistatic implementations
address irregular reference sampling, changing precoders
\cite{nataraja_bistatic}, and synchronization \cite{brunner_bistatic}.
Spatial/temporal coherence tests path consistency; departure--arrival
consistency discriminates multipath ghosts \cite{zheng_ghost}. Consistent
reflections can nevertheless produce incorrect target hypotheses.
Statistical models, ray tracing, and radio digital-twin modeling describe
complementary propagation structures \cite{zhang_channel,bazzi_rt,zhang_rf3dgs}.
These representations also connect sensing, localization, and tracking measurements to physical-state models for wireless digital twins~\cite{lin_digital_twin}.

Cooperative processing forms the second foundation. Detection combines clutter-aware observations and resource allocation \cite{behdad_detection,perovic_rate}; hybrid fusion combines monostatic and bistatic returns \cite{chowdary_hybrid}. Distributed multiple-input multiple-output (MIMO) OFDM systems can operate without carrier-phase synchronization \cite{han_noncoherent}; scalable receivers aggregate clutter-whitened detection energies \cite{liesegang_scalable}. Parameter fusion requires correspondence: Shi and Liu jointly identify line-of-sight (LOS) observations and associate targets \cite{shi_los}. Tang \emph{et al.}'s closely related UAV study combines unknown-target count estimation, tensor processing, false-removing association, three-dimensional (3D) position/velocity estimation, and dual polarization \cite{tang_uav}. Range--Doppler spatial association establishes strict groups before relaxing residual search \cite{saga}. We examine retained weak-report influence on confirmed estimates and recovery of additional objects.

Geometry and nominal precision do not settle this decision. Equivalent information characterizes localization geometry and reference nuisance \cite{shen_efi1,shen_efi2}; sensor selection optimizes covariance criteria, including correlated noise \cite{joshi_selection,chepuri_selection,liu_correlated}. Mean-square error need not inherit a log-determinant surrogate's supermodularity \cite{chamon_selection}; random uncertainty does not describe model offsets \cite{eldar_biased,fortunati_misspec}. Established minimax and weighted-risk methods address explicit uncertainty classes \cite{eldar_robust_mse,eldar_universal_weighted}; the clipping optimizer is also established \cite{jacobsson_squared_infinity}. We specialize these methods to correlated weak-report innovations, preserve confirmed associations, and regulate influence with a declared bias radius. We also identify reflected-path ambiguities that additional nodes or signal energy cannot resolve from these reports.

Reporting couples estimation with resource use. Censoring supports selective and feedback-triggered transmissions \cite{appad_censor,jiang_censor}; networked sensing optimizes compression, transmission, and beamforming \cite{zhu_compression,zhang_fronthaul}. Accessing retained candidates requires requests, so fewer uploads can increase total traffic. Our framework links observability and biased refinement with exclusive association and candidate-specific access. Strong reports confirm associations, additional views refine confirmed states under a declared risk model, and unused anchors recover objects (Fig.~\ref{fig:architecture}). The contributions are:
\begin{itemize}
\item A general observation/report model connects monostatic/bistatic geometry, nuisance references, diagnostics, and 3D observability, characterizing bias-aware distinguishability and a reflected-path equivalence class unresolved by report consistency.
\item A weighted innovation-risk formulation separates random error from unknown bias, accounts for correlation, and derives directional weak-report refinement through established minimax and proximal methods.
\item An admissible-group/access model separates confirmation, refinement, and recovery, specifying candidate requests, exact residual-search equivalence conditions, and bidirectional acquisition cost.
\item Controlled geometry, bias, correlation, and signaling experiments precede large stochastic and mapped-city evaluations. Shared-observation controls compare local rejection, association, and refinement; set-error metrics complement detection and common-target position/velocity accuracy \cite{ospa,gospa}.
\end{itemize}

The rest of this paper is organized as follows. Section~\ref{sec:model} establishes signal and report models; Section~\ref{sec:identifiability} analyzes identifiability and ambiguity; Section~\ref{sec:minimax} develops bias-aware refinement. Section~\ref{sec:algorithms} presents receiver processing, association, and reporting; Section~\ref{sec:evaluation} describes channel construction and evaluation. Section~VII concludes.

\textbf{\textit{Notation:}} Bold lowercase/uppercase letters denote vectors/matrices.
Superscripts $\mathsf T$, $\mathsf H$, and $\dagger$ denote transpose,
Hermitian transpose, and Moore--Penrose pseudoinverse; $\mathbf I$ is the
identity. Operators $\operatorname{tr}$, $\operatorname{rank}$, and
$\operatorname{diag}$ denote trace, rank, and diagonal construction;
$\mathbb E$ and $\operatorname{cov}$ denote expectation and covariance.
The Euclidean vector and spectral matrix norms are $\|\cdot\|$ and
$\|\cdot\|_2$. For vector $\mathbf e$ and weight
$\mathbf W\succeq0$,
$\|\mathbf e\|_{\mathbf W}^2=\mathbf e^{\mathsf T}\mathbf W\mathbf e$.
Symbols $\succeq0$/$\succ0$ mean positive semidefinite/definite;
matrix square roots are symmetric nonnegative roots.
The real/complex fields are $\mathbb R$/$\mathbb C$, and
$\mathrm j^2=-1$.

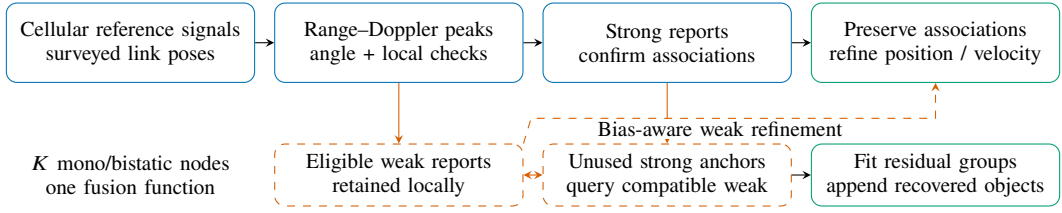
\begin{figure*}[t]\centering
\begin{tikzpicture}[font=\footnotesize,>=stealth,box/.style={draw=paperblue,rounded corners,align=center,minimum height=1.05cm,text width=3.05cm},weak/.style={draw=paperred,dashed,rounded corners,align=center,text width=3.05cm,minimum height=.82cm}]
\node[box] (p) at (0,0) {Cellular reference signals\\surveyed link poses};
\node[box] (r) at (3.55,0) {Range--Doppler peaks\\angle + local checks};
\node[box] (s) at (7.1,0) {Strong reports\\confirm associations};
\node[box,draw=papergreen] (f) at (10.65,0) {Preserve associations\\refine position / velocity};
\draw[->] (p)--(r);\draw[->] (r)--(s);\draw[->] (s)--(f);
\node[weak] (w) at (3.55,-1.75) {Eligible weak reports\\retained locally};
\node[weak] (a) at (7.1,-1.75) {Unused strong anchors\\query compatible weak};
\node[box,draw=papergreen,minimum height=.82cm] (b) at (10.65,-1.75) {Fit residual groups\\append recovered objects};
\draw[->,paperred] (r)--(w);\draw[->,paperred] (s)--(a);
\draw[<->,paperred] (w)--(a);\draw[->] (a)--(b);
\draw[->,paperred,dashed] (w.north east)--(5.3,-1)--(10.65,-1)--(f.south);
\node[font=\footnotesize,fill=white,inner sep=1pt] at (7.8,-1.14) {Bias-aware weak refinement};
\node[align=center,font=\footnotesize] at (0,-1.75) {$K$ mono/bistatic nodes\\one fusion function};
\end{tikzpicture}
\caption{Three evidence decisions: strong reports confirm associations; unconsumed weak views refine confirmed states under a declared bias risk; unused anchors recover additional objects. Confirmation identifiers remain fixed during refinement. Infinite bias radius locks the states. Truth and traced path labels enter scoring only.}
\label{fig:architecture}\end{figure*}
\section{General Signal and Measurement Model}
\label{sec:model}
\subsection{Nodes, Motion, and Reference Signals}
Consider $K$ surveyed nodes with monostatic/bistatic links and an unknown
finite target population. For an association hypothesis, target position
$\mathbf p_0\in\mathbb R^3$ at reference epoch $t_0$ and constant velocity
$\mathbf v\in\mathbb R^3$ form the state
$\mathbf x=(\mathbf p_0^{\mathsf T},\mathbf v^{\mathsf T})^{\mathsf T}\in\mathbb R^6$.
During estimation, $\mathbf p(t)=\mathbf p_0+(t-t_0)\mathbf v$ and
$\|\mathbf v\|\le v_{\max}$ for specified maximum speed $v_{\max}$;
all three velocity components are unrestricted. Link $a$ pairs
a transmitter (TX) and receiver (RX) with surveyed positions
$\mathbf b_{T,a}(t),\mathbf b_{R,a}(t)$ and velocities
$\mathbf v_{T,a},\mathbf v_{R,a}$; subscripts $T,R$ identify endpoints.
Array orientations are known; user equipment (UE) may receive.
The estimator receives neither target count nor path classification.

Known cellular reference symbols $X_{a,b}[k,n]$ sample frequencies $f_k$
and symbol times $t_n$, indexed by $k,n$.
Link $a$ has $B_a$ resources; resource $b\in\{1,\ldots,B_a\}$ uses a known
spatial weight vector $\mathbf w_{a,b}$ and transmit polarization.
Receive element/polarization index $p$ observes
\begin{equation}
\begin{aligned}
Y_{a,b,p}[k,n]&=X_{a,b}[k,n]H_{a,b,p}[k,n]+N_{a,b,p}[k,n],\\
H_{a,b,p}[k,n]&=H^{\rm bg}_{a,b,p}[k]
+\sum_{j,q}\alpha_{a,b,p,j,q}\\
&\quad\cdot
e^{-\mathrm j2\pi f_k\tau_{a,j,q}}
e^{\mathrm j2\pi f_{D,a,j,q}t_n}.
\end{aligned}\label{eq:signal}
\end{equation}
Here $Y,H,N$ are received sample, complex channel response, and thermal-noise/
post-cancellation interference sample. Stationary background $H^{\rm bg}$
excludes target interactions. Each target $j$/interacting path $q$ has
complex amplitude $\alpha_{a,b,p,j,q}$, delay $\tau_{a,j,q}$, and Doppler
$f_{D,a,j,q}$. Amplitudes include illumination, arrays, scattering, and
polarization. Links require timing/frequency registration for geometric
reports, without inter-node carrier-phase coherence.
Equation~\eqref{eq:signal} uses a local delay/Doppler approximation within each observation window.

For an isolated return of known shape with one unknown common complex amplitude
and independent circular complex Gaussian noise of equal variance, the
$2\times2$ efficient delay/Doppler Fisher information matrix
$\mathbf J_{\tau,f_D}$ is
\begin{equation}
\mathbf J_{\tau,f_D}=8\pi^2\gamma
\begin{bmatrix}
\operatorname{var}_{\omega}(f)&-\operatorname{cov}_{\omega}(f,t)\\
-\operatorname{cov}_{\omega}(f,t)&\operatorname{var}_{\omega}(t)
\end{bmatrix},\label{eq:waveform_information}
\end{equation}
where $\gamma$ is received signal energy divided by complex-noise
variance, $\omega$ are normalized sample-energy weights, and
$\operatorname{var}_{\omega},\operatorname{cov}_{\omega}$ are variance
and covariance under those weights. Removing the unknown amplitude centers
derivatives $(-\mathrm j2\pi f,\mathrm j2\pi t)$ and yields
\eqref{eq:waveform_information}. Bandwidth, timing, and their coupling determine precision;
sample count alone does not~\cite{sturm_ofdm,li_sidelobe,du_constellation}.
This motivates report weights, not direct-path classification.

\subsection{Monostatic and Bistatic Report Geometry}
At report time $t_a$, let $\Delta_a=t_a-t_0$ be the reference-epoch offset.
For endpoint $s\in\{T,R\}$, define its target displacement
$\mathbf d_{s,a}=\mathbf p_0+\Delta_a\mathbf v-\mathbf b_{s,a}(t_a)$,
distance $R_{s,a}=\|\mathbf d_{s,a}\|>0$ for noncoincident endpoints, and outward unit direction
$\mathbf u_{s,a}=\mathbf d_{s,a}/R_{s,a}$.
The direct-return prediction
$\mathbf h_a(\mathbf x)=(\ell_a,\alpha_a,\beta_a,\chi_a)^{\mathsf T}$
contains path length $\ell_a$, receive azimuth $\alpha_a$, receive elevation
$\beta_a$, and path-length rate $\chi_a$:
\begin{equation}
\mathbf h_a(\mathbf x)=
\begin{bmatrix}
R_{T,a}+R_{R,a}\\
\operatorname{atan2}(d_{R,a,y},d_{R,a,x})\\
\operatorname{atan2}(d_{R,a,z},\sqrt{d_{R,a,x}^2+d_{R,a,y}^2})\\
\mathbf u_{T,a}^{\mathsf T}(\mathbf v-\mathbf v_{T,a})
+\mathbf u_{R,a}^{\mathsf T}(\mathbf v-\mathbf v_{R,a})
\end{bmatrix}.\label{eq:geometry}
\end{equation}
Cartesian subscripts identify components of $\mathbf d_{R,a}$;
$\operatorname{atan2}(y,x)$ is the four-quadrant angle.
Path length $\ell=c_0\tau$ and rate $\chi=-\lambda f_D$ use propagation
speed $c_0$ and carrier wavelength $\lambda$; positive rate means recession.
Fixed monostatic nodes report range $r=\ell/2$ and radial speed $\nu=\chi/2$,
with uncertainty transformed accordingly. Angles use the common surveyed frame.

For a proposed target group with $m$ scalar coordinates, stack measured
reports $\mathbf z\in\mathbb R^m$ and predictions
$\mathbf h(\mathbf x)\in\mathbb R^m$ as
\begin{equation}
\mathbf z=\mathbf h(\mathbf x)+\mathbf F\boldsymbol\eta
+\mathbf b+\boldsymbol\epsilon,
\quad \mathbb E\boldsymbol\epsilon=0,
\quad\operatorname{cov}(\boldsymbol\epsilon)=\mathbf R\succ0.
\label{eq:report_model}
\end{equation}
Unknown calibration/synchronization parameters $\boldsymbol\eta$ have known
sensitivity $\mathbf F$; $\mathbf b$ is deterministic path/model bias and
$\boldsymbol\epsilon$ random error with covariance $\mathbf R$.
Cross-report blocks permit shared-receiver, waveform, and interference
correlations; Gaussian angle errors use a common local unwrapped chart.
Same-clock monostatic measurements cancel absolute clock
offsets; bistatic relative clock/frequency errors need not cancel
~\cite{brunner_bistatic,nataraja_bistatic}. Known unbiased quantization contributes
to $\mathbf R$; bounded quantization and nonlinear remainder to $\mathbf b$.

The report Jacobian $\mathbf H_a=\partial\mathbf h_a/\partial\mathbf x$
is explicit, and $\mathbf H$ stacks the link Jacobians.
Suppress index $a$ and let
$\mathbf w=\mathbf u_T+\mathbf u_R$,
$\boldsymbol\Pi_s=\mathbf I-\mathbf u_s\mathbf u_s^{\mathsf T}$, and
$\mathbf g=\boldsymbol\Pi_T(\mathbf v-\mathbf v_T)/R_T+
\boldsymbol\Pi_R(\mathbf v-\mathbf v_R)/R_R$.
With horizontal receive distance $\varrho=\sqrt{d_{R,x}^2+d_{R,y}^2}$,
\begin{align}
\mathbf a_\alpha&=(-d_{R,y}/\varrho^2,d_{R,x}/\varrho^2,0)^{\mathsf T},\nonumber\\
\mathbf a_\beta&=(-d_{R,x}d_{R,z}/(R_R^2\varrho),
-d_{R,y}d_{R,z}/(R_R^2\varrho),\varrho/R_R^2)^{\mathsf T},\nonumber\\
\mathbf H_a&=\begin{bmatrix}
\mathbf D_a&\Delta_a\mathbf D_a\\
\mathbf g^{\mathsf T}&\mathbf w^{\mathsf T}+\Delta_a\mathbf g^{\mathsf T}
\end{bmatrix},\quad
\mathbf D_a=\begin{bmatrix}\mathbf w^{\mathsf T}\\
\mathbf a_\alpha^{\mathsf T}\\\mathbf a_\beta^{\mathsf T}\end{bmatrix}.
\label{eq:jacobian}
\end{align}
Here $\boldsymbol\Pi_s$ projects perpendicular to $\mathbf u_s$;
$\mathbf a_\alpha,\mathbf a_\beta$ and $\mathbf g$ are angle and path-rate
position gradients, while $\mathbf D_a$ stacks path-length/angle gradients.
At zenith a nonsingular local direction coordinate replaces azimuth.
\section{Cooperative Identifiability and Propagation Ambiguity}
\label{sec:identifiability}
\subsection{Geometry, Clock Nuisance, and Bias}
For fixed path/bias and Gaussian errors with state-independent covariance, whiten
$\widetilde{\mathbf H}=\mathbf R^{-1/2}\mathbf H$ and
$\widetilde{\mathbf F}=\mathbf R^{-1/2}\mathbf F$.
Nuisance elimination gives equivalent Fisher information
\begin{equation}
\mathbf J_{\rm eff}=\widetilde{\mathbf H}^{\mathsf T}
(\mathbf I-\widetilde{\mathbf F}\widetilde{\mathbf F}^{\dagger})
\widetilde{\mathbf H}.\label{eq:efim}
\end{equation}
Full local observability requires rank-six $\mathbf J_{\rm eff}$ after this
established projection~\cite{shen_efi1,shen_efi2}.
For simultaneous fixed monostatic links with registered
references, range/angle rows fix
the three position perturbations. Writing
$\mathbf u_a=\mathbf u_{T,a}=\mathbf u_{R,a}$ for their common unit
direction gives
\begin{equation}
\operatorname{rank}\mathbf H=3+\operatorname{rank}\mathbf U,
\qquad \mathbf U=\operatorname{stack}_a(\mathbf u_a^{\mathsf T}).
\label{eq:rank}
\end{equation}
Here $\operatorname{stack}_a$ vertically stacks link rows; null-space
velocity perturbations $\delta\mathbf v$ satisfy $\mathbf U\delta\mathbf v=0$.
Two views cannot identify unrestricted simultaneous three-dimensional
velocity; three spanning views can. Multiple epochs add
$\Delta_a\mathbf D_a$ sensitivity. Unresolved independent path-rate offsets
can absorb all path-rate equations, irrespective of node count, so geometry
and reference registration must be reported together.

Bias also affects path distinguishability. Let
$\mathbf B\succeq0$ describe a bias ellipsoid of radius $\rho\ge0$ and
$\mathbf V$ have orthonormal columns spanning the orthogonal complement
of the nuisance column space $\operatorname{col}(\widetilde{\mathbf F})$.
For competing states $\mathbf x_j$, $j\in\{0,1\}$, define projected means
$\boldsymbol\mu_j=\mathbf V^{\mathsf T}\mathbf R^{-1/2}\mathbf h(\mathbf x_j)$
and the projected bias map
$\mathbf E=\mathbf V^{\mathsf T}\mathbf R^{-1/2}\mathbf B^{1/2}$.
Convex Gaussian testing~\cite{guigues_testing} yields the following criterion.

\begin{theorem}[Bias-aware report distinguishability]
\label{thm:distinguish}
For two fixed states under \eqref{eq:report_model}, common Gaussian covariance,
unrestricted nuisance, and independently adversarial biases
$\mathbf b_j=\mathbf B^{1/2}\boldsymbol\xi_j$, $\mathbf B\succeq0$,
$\|\boldsymbol\xi_j\|\le\rho$, the exact minimax sum of the two worst
decision-error probabilities is $2\Phi(-d_*/2)$, where $\Phi$ is the standard
normal distribution function. The minimum whitened separation is
\begin{equation}
d_* =\min_{\|\boldsymbol\delta\|\le2\rho}
\|\boldsymbol\mu_1-\boldsymbol\mu_0+\mathbf E\boldsymbol\delta\|.
\label{eq:separation}
\end{equation}
with $\boldsymbol\delta=\boldsymbol\xi_1-\boldsymbol\xi_0$ the difference
between the two admissible bias coordinates.
\end{theorem}
\begin{proof}
Whitening and nuisance projection leave identity-covariance noise and convex
mean sets $\mathcal C_j=\boldsymbol\mu_j+\mathbf E\mathbb B_\rho$, where
$\mathbb B_\rho=\{\boldsymbol\xi:\|\boldsymbol\xi\|\le\rho\}$.
If they intersect, the common mean gives a lower bound of one on the error
sum, attained by a constant decision. Otherwise let $\mathbf q_0,\mathbf q_1$
be closest points, distance $d_*$. Their separating normal
$(\mathbf q_1-\mathbf q_0)/d_*$ supports both sets. Thresholding along this
normal at the midpoint gives worst errors $\Phi(-d_*/2)$ for each set.
The Gaussian simple hypotheses at the closest points attain the same lower
bound. Unrestricted nuisance cannot improve this minimax test because it
can align the eliminated mean coordinates. Thus both bounds coincide.
\end{proof}
Theorem~\ref{thm:distinguish} separates precision from path distinguishability.
For fixed nuisance and bias classes, overlapping projected admissible mean
sets admit the same report distribution under both states. Reducing random
error leaves the minimax error sum at one; disambiguation requires additional
information or justified narrowing of these classes.

\subsection{Spatial Signatures and a Persistent Image Ambiguity}
After receive steering and Doppler compensation at each resource center,
an isolated path with common transmit polarization has resource/polarization response
$\mathbf O_{\rm sig}=\mathbf s(\mathbf u_T)\mathbf c^{\mathsf T}$.
The signature $\mathbf s(\mathbf u)$ stacks resource responses
$s_b(\mathbf u)=\mathbf a_T(\mathbf u)^{\mathsf H}\mathbf w_b$, where
$\mathbf a_T$ is the transmit-array steering vector and $\mathbf w_b$ the
resource weight with link index suppressed. Vector $\mathbf c$ contains
unknown complex receive-polarization amplitudes, and $\mathbf o_p$ is column
$p$ of $\mathbf O_{\rm sig}$. At measured unit arrival direction
$\widehat{\mathbf u}_R$, a monostatic direct-return consistency score is
\begin{equation}
C_{\rm TX}=\frac{\sum_p|\mathbf s^{\mathsf H}\mathbf o_p|^2}
{\|\mathbf s\|^2\sum_p\|\mathbf o_p\|^2},
\quad\mathbf s=\mathbf s(\widehat{\mathbf u}_R).
\label{eq:tx_signature}
\end{equation}
This normalized score requires nonzero signature and observed energies.
With independent circular complex Gaussian noise of variance $\sigma^2$ on
each resource/polarization sample, profiling out $\mathbf c$ leaves the
Fisher information $\mathbf J_\theta$ for departure azimuth/elevation
$\boldsymbol\theta$:
\begin{equation}
\mathbf J_\theta=\frac{2\|\mathbf c\|^2}{\sigma^2}
\operatorname{Re}\{\mathbf S_\theta^{\mathsf H}
(\mathbf I-\mathbf s\mathbf s^{\mathsf H}/\|\mathbf s\|^2)
\mathbf S_\theta\},\quad
\mathbf S_\theta=\partial\mathbf s/\partial\boldsymbol\theta.
\label{eq:signature_information}
\end{equation}
Here $\operatorname{Re}$ denotes real part. Eliminating each polarization's
complex amplitude yields the projection. One resource gives zero profiled
departure information; multiple resources help only through different
normalized signatures. Collinear signatures remain indistinguishable;
separation need not increase with angular difference
(Fig.~\ref{fig:path_mechanism}). Multiple transmit polarizations require a
matrix-valued signature and its corresponding nuisance subspace. Bistatic
consistency uses predicted departure $\mathbf u_T$, which generally differs
from arrival $\mathbf u_R$.

Departure/arrival disagreement is an established multipath discriminator
~\cite{zheng_ghost}. A reciprocal reflected path can nevertheless preserve
both directions. For a plane $\mathbf n^{\mathsf T}\mathbf p=d$, let
$\mathbf n$ be its unit normal and $d$ its signed offset.
Its reflection map is $g(\mathbf p)=\mathbf O\mathbf p+\mathbf t$, with
$\mathbf O=\mathbf I-2\mathbf n\mathbf n^{\mathsf T}$ and $\mathbf t=2d\mathbf n$.
Here $\mathbf O$ is orthogonal and $\mathbf t$ a translation.
An ordered sequence of planar reflections has the same affine-isometry form.
For any number of nodes sharing that valid reciprocal image family, unfolding
each path gives the direct-report geometry of
\begin{equation}
\mathbf p'=\mathbf O\mathbf p+\mathbf t,
\qquad \mathbf v'=\mathbf O\mathbf v.
\label{eq:image_state}
\end{equation}
Unfolding preserves path length; its time derivative preserves Doppler, and
the final segment points toward the image. This remains true across epochs
while the family stays visible; primes denote the resulting image position
and velocity. A bistatic extension requires the same image map on both legs.

\begin{corollary}[A report-only ambiguity floor]
\label{cor:image_floor}
Suppose two admissible physical/path hypotheses induce identical report
distributions, including any reported nuisance amplitudes. For every state
estimator $\widehat{\mathbf x}$ and chosen state-error weight
$\mathbf W\succeq0$,
\begin{equation}
\max_{j=0,1}\mathbb E_j\|\widehat{\mathbf x}-\mathbf x_j\|_{\mathbf W}^2
\ge\tfrac14\|\mathbf x_0-\mathbf x_1\|_{\mathbf W}^2.
\label{eq:ambiguity_floor}
\end{equation}
where $\mathbb E_j$ denotes expectation under hypothesis $j$.
\end{corollary}
\begin{proof}
Under the shared law, averaging the two risks gives
$\mathbb E\|\widehat{\mathbf x}-(\mathbf x_0+\mathbf x_1)/2\|_{\mathbf W}^2
+\|\mathbf x_0-\mathbf x_1\|_{\mathbf W}^2/4$.
The maximum is at least this average; the midpoint attains equality.
\end{proof}
The image construction gives identical four-coordinate reports for
admissible states with unclassified path amplitudes, without equating
complete scene responses. Map information, independent signatures, and
visibility changes can break this equivalence; additional nodes or higher
signal-to-noise ratio (SNR) alone need not.
\begin{figure*}[t]\centering
\begin{minipage}[b]{.48\textwidth}\centering
\begin{tikzpicture}\begin{axis}[paperaxis,width=\linewidth,height=5.3cm,
xlabel={Departure--arrival azimuth difference ($^\circ$)},ylabel={Ideal TX consistency},
xmin=0,xmax=60,ymin=0,ymax=1.05,legend style={at={(.02,.03)},anchor=south west,font=\scriptsize}]
\addplot[paperblue,thick] coordinates {(0,1) (0.5,0.9998784) (1,0.99975125) (1.5,0.99994783) (2,1) (2.5,0.99996572) (3,0.99996217) (3.5,0.99999143) (4,1) (4.5,0.99998897) (5,0.99994236) (5.5,0.99982703) (6,0.9995933) (6.5,0.99917733) (7,0.99866031) (7.5,0.99852258) (8,0.9981092) (8.5,0.99724875) (9,0.99676127) (9.5,0.99621724) (10,0.99494786) (10.5,0.99288079) (11,0.99183884) (11.5,0.9897992) (12,0.98666884) (12.5,0.98237918) (13,0.97689034) (13.5,0.97019429) (14,0.96231667) (14.5,0.95331656) (15,0.94328403) (15.5,0.93233536) (16,0.92060587) (16.5,0.90824096) (17,0.8953859) (17.5,0.88217536) (18,0.86872344) (18.5,0.85758721) (19,0.84846812) (19.5,0.83981213) (20,0.83158139) (20.5,0.82368866) (21,0.81600175) (21.5,0.80835026) (22,0.80053392) (22.5,0.79233199) (23,0.78351299) (23.5,0.77384443) (24,0.76898476) (24.5,0.77519697) (25,0.77991722) (25.5,0.78282507) (26,0.78362864) (26.5,0.78206925) (27,0.77792503) (27.5,0.77101361) (28,0.76119399) (28.5,0.74836764) (29,0.73247908) (29.5,0.71351612) (30,0.69150988) (30.5,0.66653487) (31,0.63870934) (31.5,0.60819591) (32,0.57520265) (32.5,0.53998462) (33,0.50284568) (33.5,0.46414033) (34,0.42427533) (34.5,0.38371033) (35,0.34295717) (35.5,0.30257681) (36,0.26317333) (36.5,0.2253842) (37,0.1898663) (37.5,0.1572776) (38,0.14631651) (38.5,0.16403307) (39,0.18956558) (39.5,0.21922107) (40,0.25243586) (40.5,0.28853287) (41,0.32674712) (41.5,0.36625784) (42,0.40622451) (42.5,0.4458238) (43,0.48428406) (43.5,0.52091472) (44,0.55512867) (44.5,0.58645646) (45,0.61455249) (45.5,0.63919387) (46,0.66027324) (46.5,0.67778732) (47,0.69182267) (47.5,0.70254021) (48,0.71015961) (48.5,0.71494448) (49,0.71718881) (49.5,0.71720512) (50,0.71531418) (50.5,0.71183645) (51,0.70708502) (51.5,0.70135988) (52,0.69494335) (52.5,0.69360802) (53,0.69553013) (53.5,0.69721298) (54,0.69881357) (54.5,0.70046737) (55,0.71443355) (55.5,0.72847938) (56,0.74197418) (56.5,0.75498066) (57,0.76755106) (57.5,0.77972624) (58,0.79153492) (58.5,0.80299303) (59,0.81410322) (59.5,0.82485451) (60,0.83522213)};\addlegendentry{Three implemented codes}
\addplot[paperblack,dashed] coordinates {(0,1)(60,1)};\addlegendentry{One code}
\addplot[paperred,dotted,thick] coordinates {(0,.85)(60,.85)};\addlegendentry{Strong gate 0.85}
\end{axis}\end{tikzpicture}
\par\smallskip{\footnotesize(a) Signature separation can be nonmonotone}
\end{minipage}\hfill
\begin{minipage}[b]{.48\textwidth}\centering
\begin{tikzpicture}[x=1.1cm,y=1.1cm,font=\scriptsize,>=stealth]
\path[use as bounding box] (-2.55,-1.15) rectangle (2.3,3.4);
\draw[very thick,black!60] (-2.4,1.8)--(2.1,1.8);
\node[anchor=west,fill=white,inner sep=1pt] at (.35,1.95) {stationary reflecting plane};
\node[circle,fill=paperblack,inner sep=2pt,label={[xshift=-18pt]below:UAV $\mathbf p,\mathbf v$}] at (0,.8) (p) {};
\node[circle,draw=paperred,thick,inner sep=2pt,label={above:virtual $\mathbf p',\mathbf v'$}] at (0,2.8) (v) {};
\draw[gray,densely dotted] (p)--(v);
\node[rectangle,fill=paperblue,inner sep=2pt,label=below:Node 0] at (-2.0,0.0) (b0) {};
\coordinate (q0) at (-0.7142857142857142,1.8);\draw[paperblue,thick] (b0)--(q0)--(p);\draw[paperblue,dashed] (q0)--(v);
\node[rectangle,fill=papergreen,inner sep=2pt,label=below:Node 1] at (1.0,-0.8) (b1) {};
\coordinate (q1) at (0.27777777777777785,1.8);\draw[papergreen,thick] (b1)--(q1)--(p);\draw[papergreen,dashed] (q1)--(v);
\node[rectangle,fill=paperred,inner sep=2pt,label=below:Node 2] at (1.0,0.8) (b2) {};
\coordinate (q2) at (0.5,1.8);\draw[paperred,thick] (b2)--(q2)--(p);\draw[paperred,dashed] (q2)--(v);

\end{tikzpicture}
\par\smallskip{\footnotesize(b) Reciprocal shared-image ambiguity (plan view)}
\end{minipage}
\caption{Information and ambiguity of coded reference-signal paths. (a) The exact panel's departure signature is tested within a spherical $4^\circ$ cap about arrival $0^\circ/30^\circ$, with $2^\circ$ candidate spacing. Departure elevation is $30^\circ$; the path is noiseless. (b) A controlled valid plane at $y=180$ m reflects $\mathbf p=(0,80,120)$ m to $\mathbf p'=(0,280,120)$ m and $\mathbf v=(20,-12,5)$ m/s to $\mathbf v'=(20,12,5)$ m/s. Report equivalence assumes illuminated, admissible image states and isolated reciprocal paths.}
\label{fig:path_mechanism}\end{figure*}
\section{Bias-Aware Refinement of Confirmed Targets}
\label{sec:minimax}
\subsection{Risk of a Weak-Report Update}
The strong-report baseline estimate $\widehat{\mathbf x}_S$ has zero-mean
error $\mathbf e_S=\widehat{\mathbf x}_S-\mathbf x$ and covariance
$\mathbf P\succeq0$. After nuisance removal, recenter and linearize fixed
additional reports as $\mathbf z=\mathbf H\mathbf x+\mathbf b+\boldsymbol\epsilon$.
Here $\mathbf z\in\mathbb R^m$ contains $m$ scalar coordinates, $\mathbf H$
is their state-sensitivity matrix, $\mathbf b$ deterministic bias, and
$\boldsymbol\epsilon$ zero-mean error with covariance $\mathbf R$.
Cross-covariance $\mathbf C=\operatorname{cov}(\mathbf e_S,\boldsymbol\epsilon)$
describes baseline/weak-noise dependence. Define
\begin{align}
\mathbf r&=\mathbf z-\mathbf H\widehat{\mathbf x}_S,
\quad\mathbf A=\mathbf P\mathbf H^{\mathsf T}-\mathbf C,\nonumber\\
\mathbf T&=\mathbf H\mathbf P\mathbf H^{\mathsf T}+\mathbf R
-\mathbf H\mathbf C-\mathbf C^{\mathsf T}\mathbf H^{\mathsf T}\succ0.
\label{eq:innovation}
\end{align}
Innovation $\mathbf r$ subtracts the baseline prediction; $\mathbf T$ is
its random covariance and $\mathbf A=-\operatorname{cov}(\mathbf e_S,\mathbf r)$
the signed baseline coupling. The joint covariance of
$(\mathbf e_S,\boldsymbol\epsilon)$ must be positive semidefinite.
Bias is $\mathbf b=\mathbf B^{1/2}\boldsymbol\xi$, with ellipsoid shape
$\mathbf B\succeq0$, unknown normalized vector $\boldsymbol\xi$, and radius
$\rho\ge0$ satisfying $\|\boldsymbol\xi\|\le\rho$. Coverage of path errors,
quantization, and nonlinear remainders requires separate justification.

Gain $\mathbf L$ maps innovations to state corrections; loss matrix
$\mathbf W\succeq0$ balances units and application priorities. Estimate
$\widehat{\mathbf x}_L=\widehat{\mathbf x}_S+\mathbf L\mathbf r$ has worst
expected weighted squared error
$\mathcal R_\rho(\mathbf L)=\sup_{\|\boldsymbol\xi\|\le\rho}
\mathbb E\|\widehat{\mathbf x}_L-\mathbf x\|_{\mathbf W}^2$, namely
\begin{align}
\mathcal R_\rho(\mathbf L)
&=\operatorname{tr}(\mathbf W\mathbf P)
-2\operatorname{tr}(\mathbf W\mathbf L\mathbf A^{\mathsf T})
+\operatorname{tr}(\mathbf W\mathbf L\mathbf T\mathbf L^{\mathsf T})\nonumber\\
&\quad+\rho^2\|\mathbf W^{1/2}\mathbf L\mathbf B^{1/2}\|_2^2.
\label{eq:minimax_risk}
\end{align}
Expanding random-error loss and maximizing deterministic bias gives
\eqref{eq:minimax_risk} without Gaussianity. Hypothesis, gain, moments, and
bias set are fixed during expectation, or conditioned on independent
calibration; data-dependent association and covariance proxies require
empirical validation. Classical robust estimation separates variance from
bias \cite{kassam_robust,eldar_biased,fortunati_misspec,eldar_robust_mse}.
The following minimax specialization allocates least-favorable bias energy
to correlated strong/weak innovations under a specified sensing bias class.

\begin{theorem}[Exact multi-report minimax refinement]
\label{thm:minimax}
Let $\mathcal Q_\rho=\{\mathbf Q\succeq0:\operatorname{tr}\mathbf Q\le\rho^2\}$
be the admissible set of bias-energy allocation matrices, and let
$\mathbf S_Q=\mathbf T+\mathbf B^{1/2}\mathbf Q\mathbf B^{1/2}$ be the
corresponding augmented innovation matrix. The convex problem defining the
optimized risk reduction $V_\rho$,
\begin{equation}
V_\rho=\min_{\mathbf Q\in\mathcal Q_\rho}
\operatorname{tr}(\mathbf W\mathbf A\mathbf S_Q^{-1}\mathbf A^{\mathsf T})
\label{eq:dual}
\end{equation}
attains its solution. Starred symbols denote optimizers. For a minimizer
$\mathbf Q_*$,
\begin{equation}
\mathbf L_*=\mathbf A\mathbf S_{Q_*}^{-1},\qquad
\min_{\mathbf L}\mathcal R_\rho(\mathbf L)
=\operatorname{tr}(\mathbf W\mathbf P)-V_\rho.
\label{eq:optimal_gain}
\end{equation}
The weighted gain $\mathbf W^{1/2}\mathbf L_*$ is unique. Improvement over
no update is strict for finite $\rho$ exactly when $\mathbf W^{1/2}\mathbf A\ne0$.
\end{theorem}
\begin{proof}
Represent the bias term as
$\max_{\mathbf Q\in\mathcal Q_\rho}
\operatorname{tr}(\mathbf W\mathbf L\mathbf B^{1/2}\mathbf Q
\mathbf B^{1/2}\mathbf L^{\mathsf T})$.
For fixed $\mathbf Q$, completing the square gives minimizer
$\mathbf L_Q=\mathbf A\mathbf S_Q^{-1}$ and value
$\operatorname{tr}(\mathbf W\mathbf P)-
\operatorname{tr}(\mathbf W\mathbf A\mathbf S_Q^{-1}\mathbf A^{\mathsf T})$.
The inverse-trace objective is convex and continuous on compact
$\mathcal Q_\rho$. Its derivative is
$-\mathbf B^{1/2}\mathbf L_Q^{\mathsf T}\mathbf W\mathbf L_Q\mathbf B^{1/2}$.
First-order optimality makes $\mathbf Q_*$ maximize bias at $\mathbf L_*$,
which minimizes the fixed-$\mathbf Q_*$ quadratic. These saddle inequalities prove
\eqref{eq:optimal_gain}. Positive-definite $\mathbf T$ gives strict convexity
in the weighted gain. If $\mathbf W^{1/2}\mathbf A\ne0$, the positive
inverse-trace objective has a positive minimum on the compact set; otherwise
the value is zero. This proves the improvement criterion.
\end{proof}
\subsection{Directional Admission and Spectral Refinement}
Take the innovation-normalized design set
$\mathbf b^{\mathsf T}\mathbf T^{-1}\mathbf b\le\rho^2$, i.e., $\mathbf B=\mathbf T$.
The resulting mode matrix has eigendecomposition
\begin{equation}
\mathbf M=\mathbf T^{-1/2}\mathbf A^{\mathsf T}\mathbf W\mathbf A\mathbf T^{-1/2}
=\mathbf V_M\operatorname{diag}(\mu_j)\mathbf V_M^{\mathsf T},
\label{eq:gain_spectrum}
\end{equation}
Here $\mathbf V_M$ contains orthonormal eigenvectors, $\mu_j\ge0$ are
mode eigenvalues ($j=1,\ldots,m$), and $(x)_+=\max(x,0)$ denotes positive part.

\begin{corollary}[Spectral admission and its scalar alternative]
\label{cor:cap}
For $\rho>0$ and $\mathbf M\ne0$, a unique spectral cap $\kappa>0$ solves
\begin{equation}
\rho^2\kappa=\sum_j(\sqrt{\mu_j}-\kappa)_+.
\label{eq:cap}
\end{equation}
The admission factor of mode $j$ is
$s_j=\min(1,\kappa/\sqrt{\mu_j})$, with $s_j=1$ for $\mu_j=0$.
Then
\begin{equation}
\mathbf L_*=\mathbf A\mathbf T^{-1/2}\mathbf V_M
\operatorname{diag}(s_j)\mathbf V_M^{\mathsf T}\mathbf T^{-1/2},
\quad V_\rho=\sum_j\mu_js_j.
\label{eq:spectral_gain}
\end{equation}
Restricting $\mathbf L=\alpha\mathbf A\mathbf T^{-1}$ uses a uniform scalar
gain coefficient $\alpha$. Its optimum is
\begin{equation}
\alpha_*={\operatorname{tr}\mathbf M\over
\operatorname{tr}\mathbf M+\rho^2\lambda_{\max}(\mathbf M)},
\label{eq:uniform_gain}
\end{equation}
where $\lambda_{\max}(\mathbf M)$ is the largest eigenvalue of $\mathbf M$.
The directional optimum is strictly better than this optimized scalar rule
when the nonzero $\mu_j$ are unequal, and equal when they are all equal.
\end{corollary}
\begin{proof}
In the eigenbasis of $\mathbf M$, the Schur-complement inequality
$[(\mathbf I+\mathbf Q)^{-1}]_{jj}\ge1/(1+Q_{jj})$ shows that a diagonal
$\mathbf Q$ is optimal in \eqref{eq:dual}. Writing $q_j=Q_{jj}$ for the
diagonal allocation to mode $j$, its scalar problem is
$\min_{q_j\ge0,\sum q_j\le\rho^2}\sum_j\mu_j/(1+q_j)$.
The Karush--Kuhn--Tucker optimality conditions yield
$q_j=(\sqrt{\mu_j}/\kappa-1)_+$ and a tight energy
constraint, giving \eqref{eq:cap}--\eqref{eq:spectral_gain}.
For a scalar gain the risk reduction is
$2\alpha a-\alpha^2(a+\rho^2b)$, where
$a=\operatorname{tr}\mathbf M$ and $b=\lambda_{\max}(\mathbf M)$;
maximizing gives \eqref{eq:uniform_gain}. Strict convexity makes equality
possible only when every useful mode has the same shrinkage. The cap law
has that property exactly when the positive eigenvalues coincide.
\end{proof}
The rule specializes established singular-value clipping
\cite{jacobsson_squared_infinity} to innovation gains. A six-state target has
at most six useful modes; sorting them or bisecting \eqref{eq:cap} computes
the update after whitening.

As $\rho\to\infty$, full-rank bias uncertainty suppresses the weighted update
(the complete update when $\mathbf W\succ0$). More generally, with
$\mathbf E=\mathbf T^{-1/2}\mathbf B^{1/2}$ denoting the whitened bias
directions,
\begin{equation}
\mathbf L_\infty=\mathbf A\mathbf T^{-1/2}
(\mathbf I-\mathbf E\mathbf E^{\dagger})\mathbf T^{-1/2}.
\label{eq:bias_projection}
\end{equation}
The weighted gain converges to $\mathbf W^{1/2}\mathbf L_\infty$: bounded risk
forces weighted bias response to zero, while the positive-definite quadratic
bounds the weighted gain. Every limit solves the bias-free constrained
quadratic, whose unique weighted solution is \eqref{eq:bias_projection}.
The factor $\mathbf I-\mathbf E\mathbf E^{\dagger}$ projects away the
column space of those bias directions, so bias-free report combinations
survive. Before optimization, unbounded
reference nuisance uses the same projection with its nuisance directions.

\subsection{Joint Report Selection}
For a fixed baseline and consistently projected global bias set, subset
$\mathcal J$ has optimized reduction $V_\rho(\mathcal J)$. Monotonicity follows
by allowing zero gain on extra reports. With acquisition cost $c_a\ge0$ for
report $a$ and communication budget $C_{\max}$, selection can choose
\begin{equation}
\max_{\mathcal J:\sum_{a\in\mathcal J}c_a\le C_{\max}}
V_\rho(\mathcal J).\label{eq:budget}
\end{equation}
Fisher-based selection has substantial precedents~\cite{joshi_selection,
chepuri_selection,liu_correlated,chamon_selection}; robust value cannot
generally inherit a greedy guarantee. For scalar baseline error variance
$P=1$, two
unit-sensitivity reports, $\mathbf C=0$, and
$\mathbf R=\left[\begin{smallmatrix}100&-99\\-99&100\end{smallmatrix}\right]$,
the nominal reduction at $\rho=0$ is $1/101$ for either report and $2/3$ for both.
With $\mathbf B=\mathbf T$ all three values are divided by $1+\rho^2$.
Correlated-noise cancellation makes joint value exceed summed individual
values at every finite radius. Group requests can matter despite low
individual signal-to-noise ratios.

\section{Reference-Signal Receiver, Evidence Decisions, and Reporting}
\label{sec:algorithms}
\subsection{Local Processing and Validity Tiers}
Each receiver removes the known reference symbols, projects each tone/antenna time sequence
onto the complement of its constant component, and applies Hann-tapered
range--Doppler transforms. Mean projection cancels the modeled stationary
background and attenuates near-zero radial-speed targets. Peaks are thresholded
using a Gamma-noise family threshold, refined in range/angle, and subjected
to seven consistency checks. The threshold budgets false alarms over search
cells, without determining multipath validity.
Spatial coherence $C_s$ is normalized energy explained by one receive steering
vector across resource/polarization groups. Temporal coherence $C_t$ is energy
explained by one Doppler phase progression. A noise-debiased implementation
of \eqref{eq:tx_signature}, maximized within a small arrival-angle cap,
checks whether resource amplitudes agree with direct illumination.
The coherent processing interval (CPI)
is the observation interval over which Doppler phase is fitted; estimates
from its first and second halves are
compared in direction, unfolded speed, and power. Finally, a weaker delayed
peak near a stronger angle/speed-compatible return is rejected as a possible
companion. Each test measures consistency, not a direct-path probability.
Factorial ablation evaluates all $2^7$ check coalitions under the fixed
processing order; Shapley-weighted marginal reductions allocate interaction
credit without adding isolated gains.

For node $i\in\{1,\ldots,K\}$, strong-report set $S_i$ contains candidates
that pass all local checks. Eligible-report set
$E_i\supseteq S_i$ retains weaker candidates and their diagnostics, with
$W_i=E_i\setminus S_i$. Each report includes the four measured coordinates,
an epoch, stable identifier (ID), source node/resource, and covariance or weight matrix.
The implemented diagonal standard-deviation proxy is
\begin{equation}\begin{aligned}
\sigma_\zeta&=\max\{\delta_\zeta/\sqrt{12},\ a_\zeta/\sqrt{\gamma_{\rm eff}}\},\\
\gamma_{\rm eff}&=\max\{1,\gamma\max(C_sC_t,0.05)\}.
\end{aligned}
\label{eq:proxy}
\end{equation}
Here $\zeta\in\{r,\alpha,\beta,\nu\}$ identifies range, receive azimuth,
elevation, or radial speed; $\sigma_\zeta$ is its reported standard deviation,
$\delta_\zeta$ its grid spacing, and $a_\zeta$ its waveform/aperture resolution
scale. The noise-subtracted observed SNR $\gamma$
gives the effective precision weight $\gamma_{\rm eff}$. The separate bias
model governs refinement influence.

\subsection{Association, Confirmation, and Protected Refinement}
For monostatic report $a$, apparent position is
$\widetilde{\mathbf p}_a=\mathbf b_a+r_a\widetilde{\mathbf u}_a$, where
$\mathbf b_a$ is the surveyed node position, $r_a$ the measured range, and
$\widetilde{\mathbf u}_a$ the unit vector obtained from measured angles.
Propagate the range/angle covariance through
$\mathbf J_a=[\widetilde{\mathbf u}_a,
r_a\partial_\alpha\widetilde{\mathbf u}_a,
r_a\partial_\beta\widetilde{\mathbf u}_a]$, the coordinate-conversion
Jacobian, to obtain Cartesian covariance
$\mathbf C_a=\mathbf J_a\mathbf R_{a,r\alpha\beta}\mathbf J_a^{\mathsf T}$,
where $\mathbf R_{a,r\alpha\beta}$ is the range/angle covariance submatrix
and $\partial_\alpha,\partial_\beta$ are angular derivatives.
The binary compatibility graph gives reports from different nodes an edge
$G(a,b)=1$ if
\begin{equation}\begin{gathered}
\boldsymbol\Sigma_{ab}=\mathbf C_a+\mathbf C_b-\mathbf C_{ab}
-\mathbf C_{ab}^{\mathsf T}+\mathbf C_{ab}^{\rm motion},\\
(\widetilde{\mathbf p}_a-\widetilde{\mathbf p}_b)^{\mathsf T}
\boldsymbol\Sigma_{ab}^{-1}
(\widetilde{\mathbf p}_a-\widetilde{\mathbf p}_b)\le\tau_g.
\label{eq:pair_gate}
\end{gathered}\end{equation}
Here $\tau_g$ is the threshold, $\mathbf C_{ab}$ the apparent-position
cross-covariance, and $\mathbf C_{ab}^{\rm motion}$ a positive semidefinite
allowance for the speed bound and epoch difference. Evaluated independent
proxies set $\mathbf C_{ab}=0$.
Bistatic queries instead use the transmitter/receiver geometry
in \eqref{eq:geometry}; half-delay is not a monostatic range.
Report prediction $\mathbf h_a^{\rm rep}$ equals
$\mathbf D_{\rm m}\mathbf h_a$, with
$\mathbf D_{\rm m}=\operatorname{diag}(1/2,1,1,1/2)$, for monostatic links,
and $\mathbf h_a$ for bistatic links. Monostatic covariance is
$\mathbf R_a=\mathbf D_{\rm m}\mathbf R_a^{\rm path}\mathbf D_{\rm m}^{\mathsf T}$,
where $\mathbf R_a^{\rm path}$ is path-coordinate covariance.
An admissible $\mathcal G$ requires at least two nodes, one report per node,
all pair edges, and a successful joint fit:
\begin{equation}\begin{aligned}
\widehat{\mathbf x}_{\mathcal G}
=\arg\min_{\mathbf x}\sum_{a\in\mathcal G}\sum_k\ell_{\kappa_H}(e_{a,k}),
\\\mathbf e_a=\mathbf L_a^{-1}[\mathbf z_a-\mathbf h_a^{\rm rep}(\mathbf x)],
\quad\mathbf R_a=\mathbf L_a\mathbf L_a^{\mathsf T}.
\end{aligned}
\label{eq:fit}
\end{equation}
In \eqref{eq:fit}, $\mathbf z_a$ and $\mathbf R_a$ are report $a$ and its
covariance, $\mathbf L_a$ is its lower-triangular Cholesky factor, and
$e_{a,k}$ is coordinate $k$ of the whitened residual $\mathbf e_a$.
This fit uses block-diagonal weight proxies; known cross-report covariance
instead whitens the stacked residual jointly.
Nuisance references are registered before fitting; unresolved
$\boldsymbol\eta$ must be fitted jointly or projected out.
Angle residuals are wrapped before whitening. The Huber loss
$\ell_{\kappa_H}(e)$ equals $e^2/2$ for $|e|\le\kappa_H$ and
$\kappa_H|e|-\kappa_H^2/2$ otherwise, with dimensionless threshold
$\kappa_H>0$. Acceptance also bounds
$\max|e_{a,k}|$, $Q=\sum_a\|\mathbf e_a\|^2$, and fitted speed.
The residual dimension is $m-\operatorname{rank}\widetilde{\mathbf H}$
for $m$ scalar measurement coordinates. Post-selection thresholds are design
gates, not calibrated chi-square probabilities. Unobservable velocity
directions are flagged.

Following established association foundations~\cite{reid_mht,saga}, we use one common admissible-group family
to compare strict consensus (SC), which requires all-strong support;
permissive confirmation (PC), which permits all-weak support;
anchored confirmation (AC), which requires at least one strong report; and
staged anchored recovery (SAR), which preserves SC before residual AC.
Greedy exclusive selection ranks groups by strong support, total support,
normalized residual, quality, then stable IDs. It is deterministic and need
not solve globally optimal set packing. Table~\ref{tab:algorithm} summarizes
the resulting processing sequence.

The risk-aware extension separates \emph{confirmation identity} from
\emph{state refinement}. After residual recovery it applies
\eqref{eq:spectral_gain} to weak augmentations whose strong membership
exactly matches the confirmed group, excluding records consumed by recovery.
Thus refinement preserves confirmed associations. At $\rho=0$ the gain is
nominal; infinite radius reproduces locked SAR, and finite radius controls
directional influence.

\begin{table}[t]\centering
\caption{Operational evidence processing. The fitting model remains general.}
\label{tab:algorithm}
\setlength{\tabcolsep}{3pt}
\begin{tabular}{@{}p{.035\columnwidth}p{.915\columnwidth}@{}}\toprule
1 & Upload strong reports with poses/references, IDs, epochs, and uncertainty.\\
2 & Fit compatible strong groups; select exclusive SC confirmations.\\
3 & Reserve their IDs; unused strong reports become residual anchors.\\
4 & Query eligible weak records in each anchor's compatibility neighborhood.\\
5 & Fit/select residual anchored groups; append recovered confirmations.\\
6 & Acquire eligible weak refinements around confirmed states;
    preserve strict strong membership;
    exclude records already consumed by recovery.\\
7 & Form $\mathbf P,\mathbf R,\mathbf C,\mathbf H$ and the specified bias set;
    remove unresolved reference directions; compute $\mathbf L_*$.\\
8 & Update confirmed states; report support, uncertainty assumptions, and
    observable velocity subspace.\\\bottomrule
\end{tabular}\end{table}

\subsection{Residual Access and Communication Cost}
After SC reserves its reports, let $S^{\rm res}$ contain unused strong anchors.
Let $W=\bigcup_i W_i$ denote the network's eligible weak-report pool. Request
\begin{equation}
W^{\rm cue}=\{w\in W:\exists s\in S^{\rm res},\ G(s,w)=1\}.
\label{eq:cue}
\end{equation}
With fixed gates, cached values/weights, stable IDs, complete replies, and
uncapped enumeration, screening preserves the residual anchored group family:
every weak member has an edge to its group's anchor and returns, and
screening adds no group. Identical deterministic fitting/selection therefore
gives identical residual outputs. This property concerns recovery rather
than one-pass AC or refinement traffic.

Let $N_w,N_q,N_a$ count eligible weak reports, distinct recovery replies,
and residual anchors. Let $b_w,b_c$ be weak-report and cue payload
sizes in bytes. Under unicast requests to $K-1$ other nodes, the staged
payload $B_{\rm SAR}$ and full-push payload $B_{\rm push}$ differ by
\begin{equation}
B_{\rm SAR}-B_{\rm push}=(K-1)N_ab_c-(N_w-N_q)b_w.
\label{eq:traffic}
\end{equation}
The strong payload cancels. For $N_a>0$, savings require
$b_c/b_w<(N_w-N_q)/[(K-1)N_a]$. Censoring, triggered reporting, and fronthaul
compression are established~\cite{viswanathan_detection,appad_censor,
jiang_censor,zhu_compression,zhang_fronthaul}; this calculation
counts the requests needed to acquire candidate-level evidence.

Fusion needs more than point estimates: both endpoint poses and resource
identity, acquisition/reference epochs and frame, delay/rate conventions and
ambiguity status, uncertainty with units and cross-report correlation, validity
tier and consistency scores, and calibration/bias-model provenance.
Covariance describes random spread; it cannot substitute for a declared
path-bias envelope. Compressed reports must retain the quantization rule so
its uncertainty can enter \eqref{eq:report_model}. The same fields support
fusion from multiple bistatic UEs using \eqref{eq:geometry}.
\begin{figure}[t]
\centering
\begin{minipage}[b]{.98\columnwidth}\centering
\\[4pt](b) Native San Francisco roof projection
\end{minipage}
\caption{Study layouts. Both use three transmission/reception points (TRPs) separated by 500 m horizontally. In (b), the native terrain places TRPs at world heights 89.25, 50.96 and 28.23 m, all 25 m above local terrain. Building roof boundaries come from the native mesh within the displayed crop; terrain is omitted from the drawing, retained by the tracer. Dashed curves delimit the surveillance sector. San Francisco geometry: Sionna ray-tracing packaged scene, derived from OpenStreetMap (Open Database License).}
\label{fig:deployment}
\end{figure}
\section{Channel Construction and Evaluation}
\label{sec:channel}\label{sec:evaluation}
Both environments generate the fields in \eqref{eq:signal} for one common
receiver: stochastic channels capture diffuse target paths, while mapped
ray tracing retains persistent reflector geometry~\cite{zhang_channel,bazzi_rt,
zhang_rf3dgs}.

\subsection{3GPP Stochastic Channel Construction}
The Third Generation Partnership Project (3GPP) environment combines
urban macrocell (UMa) aerial-vehicle (UMa-AV) one-way LOS
probability, pathloss, shadowing, and a 15 dB LOS Ricean factor from
3GPP Technical Report (TR)~36.777~\cite{tr36777} with outdoor UMa fast fading.
Each transmission/reception point (TRP)--UAV pair draws one LOS or non-line-of-sight
(NLOS) condition and a reciprocal channel shared by both monostatic legs;
sites have independent realizations. Sensing concatenation option~2 in
TR~38.901~\cite{tr38901} retains LOS--LOS, LOS--diffuse, diffuse--LOS, and
globally permuted diffuse--diffuse ray pairs, with 20 rays per cluster and
prescribed subdelays. LOS--LOS is always retained; other pairs are discarded
when more than 40 dB below the strongest of those other pairs.

For normalized one-way powers $p_u,p_v$, paired target power is proportional
to $p_up_v\sigma_{uv}$, with radar cross section (RCS)
\begin{equation}
\sigma_{uv}=10^{-12.81/10}S_{uv}
10^{-3\sin(\beta_{uv}/2)/10}\ {\rm m}^2.
\label{eq:rcs}
\end{equation}
RCS is scattering strength expressed as equivalent isotropic area;
$\beta_{uv}$ is the target-centered ray angle, and $S_{uv}$ is dimensionless
lognormal strength with 3.74 dB spread and unit linear mean before a
three-standard-deviation upper clip, shared by reciprocal pairs.
Polarization matrices transpose under reversal; cross-polarization power
ratio (XPR) has Gaussian dB mean 13.75 and standard deviation 7.07.
For wavelength $\lambda$, two-leg dB loss includes
$20\log_{10}\lambda-10\log_{10}(4\pi)$ before subtracting RCS.
With geometric range $R$, propagation speed $c_0$, one-way excess delays
$\delta_u,\delta_v$, and diffuse indicators $I_u,I_v$, paired delay is
$2R/c_0+\delta_u+\delta_v+(I_u+I_v)\Delta$.
The additional delay obeys
$\log_{10}(\Delta/{\rm s})\sim\mathcal N(-7.4,0.2^2)$, where
$\mathcal N$ denotes normal distribution with specified mean and variance.
Doppler follows the two outward target directions.

Three independent UMa-NLOS virtual reference channels produce stationary
background. Gamma reference distances have shape 10.337, rate
$0.1317\ {\rm m}^{-1}$, and 68.7778 m shift; heights have
shape 16.2253, rate $1.9218\ {\rm m}^{-1}$, and 2.6142 m shift. Azimuths share a uniform
random rotation and $120^\circ$ offsets; departure/arrival directions are
equal. Each channel uses one reference-link loss and virtual delay,
zero Doppler, and the three-channel sum is unnormalized.

\subsection{Mapped-City Ray-Tracing Construction}
Sionna Ray Tracing (RT)~2.2.0's packaged San Francisco scene has 4,735 objects,
268,041 triangles, hilly terrain, and retained materials.
Figure~\ref{fig:deployment} shows its OpenStreetMap-derived map and the
stochastic layout. Target tracing depth three permits one target interaction
and at most two urban reflections; stationary-background depth is also three.
Diffuse scattering, diffraction, refraction, and moving background are
disabled. Isotropic target RCS shares one 3.74 dB lognormal strength across
paths/sites, with fixed reciprocal polarization and 13.75 dB XPR.
Unlike \eqref{eq:rcs}, this construction defines a complementary benchmark,
not a propagation-only comparison. Target and background budgets are
$2^{20}$ and $10^5$ rays per call, respectively.

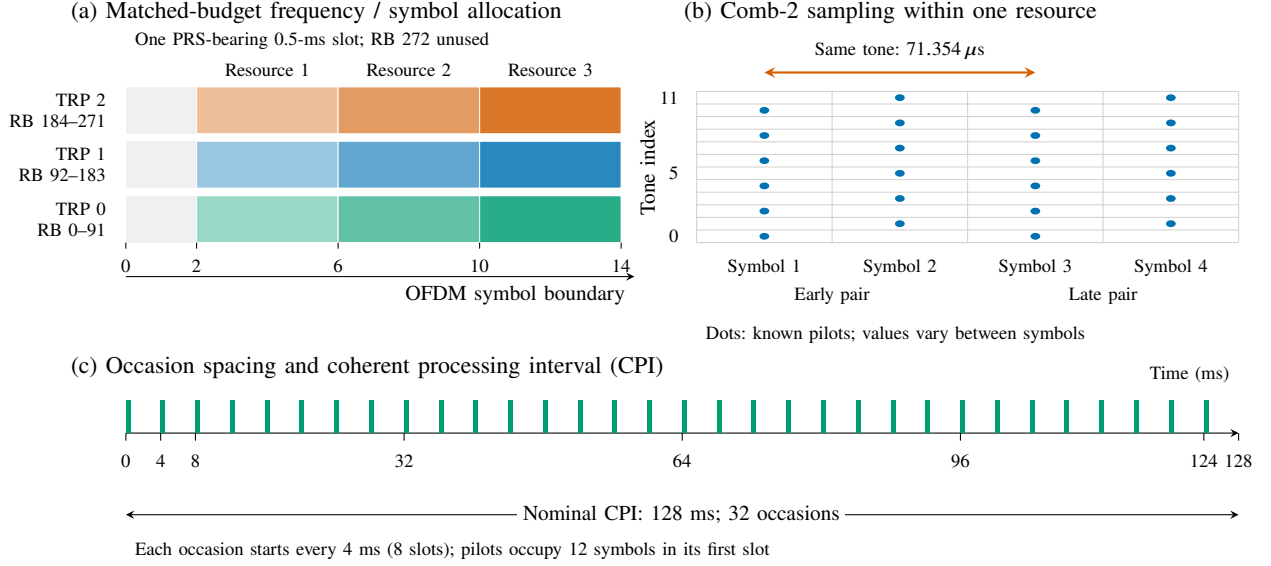
\begin{figure*}[t]\centering
\begin{tikzpicture}[font=\footnotesize,>=stealth,x=1cm,y=1cm]
\begin{scope}[xshift=-.595cm,xscale=1.17]
\node[anchor=west,font=\small] at (.124,3.05) {(a) Matched-budget frequency / symbol allocation};
\foreach \y/\col/\name/\band in {0/papergreen/0/{0--91},.72/paperblue/1/{92--183},1.44/paperred/2/{184--271}} {
  \fill[gray!12] (.85,\y) rectangle (1.65,\y+.62);
  \foreach \s/\op in {1.65/.40,3.25/.62,4.85/.84} {
    \fill[\col,opacity=\op] (\s,\y) rectangle (\s+1.6,\y+.62);
    \draw[white] (\s,\y) rectangle (\s+1.6,\y+.62);
  }
  \node[anchor=east,align=right,font=\scriptsize] at (.73,\y+.31) {TRP \name\\RB \band};
}
\foreach \x/\v in {.85/0,1.65/2,3.25/6,4.85/10,6.45/14} {
  \draw (\x,0)--(\x,-.08);
  \node[anchor=north,font=\scriptsize] at (\x,-.10) {\v};
}
\foreach \x/\v in {2.45/1,4.05/2,5.65/3}
  \node[font=\scriptsize] at (\x,2.28) {Resource \v};
\draw[->] (.85,-.45)--(6.6,-.45) node[anchor=north east] {OFDM symbol boundary};
\node[anchor=west,font=\scriptsize] at (.85,2.67) {One PRS-bearing 0.5-ms slot; RB 272 unused};
\end{scope}
\begin{scope}[xshift=7.95cm,xscale=1.434]
\node[anchor=west,font=\small] at (-.2,3.05) {(b) Comb-2 sampling within one resource};
\foreach \x in {0,1.25,2.5,3.75,5} \draw[gray!35] (\x,0)--(\x,2);
\foreach \k in {0,...,12} \draw[gray!35] (0,{\k/6})--(5,{\k/6});
\foreach \s in {0,1,2,3} {
  \foreach \k in {0,...,11} {
    \pgfmathtruncatemacro{\occupied}{mod(\s+\k,2)}
    \ifnum\occupied=0
      \fill[paperblue] (\s*1.25+.625,{(\k+.5)/6}) circle[radius=.045];
    \fi
  }
  \node[font=\scriptsize,anchor=north] at (\s*1.25+.625,-.10) {Symbol \the\numexpr\s+1\relax};
}
\foreach \k in {0,5,11} \node[font=\scriptsize,anchor=east] at (-.08,{(\k+.5)/6}) {\k};
\node[rotate=90] at (-.45,1) {Tone index};
\draw[<->,paperred,thick] (.625,2.25)--(3.125,2.25);
\node[font=\scriptsize,anchor=south] at (1.875,2.31) {Same tone: $71.354\,\mu$s};
\node[font=\scriptsize,anchor=north] at (1.25,-.47) {Early pair};
\node[font=\scriptsize,anchor=north] at (3.75,-.47) {Late pair};
\node[font=\scriptsize,anchor=north west] at (0,-.99) {Dots: known pilots; values vary between symbols};
\end{scope}
\begin{scope}[yshift=-2.85cm]
\node[anchor=west,font=\small] at (-.45,1.2) {(c) Occasion spacing and coherent processing interval (CPI)};
\draw[->] (.4,.33)--(15.12,.33);
\node[anchor=south east,font=\scriptsize] at (15.12,.88) {Time (ms)};
\foreach \o in {0,...,31} {
  \fill[papergreen] (.4+\o*.46,.34) rectangle (.4+\o*.46+.0575,.77);
}
\foreach \o/\t in {0/0,1/4,2/8,8/32,16/64,24/96,31/124,32/128} {
  \draw (.4+\o*.46,.33)--(.4+\o*.46,.23);
  \node[font=\scriptsize,anchor=north] at (.4+\o*.46,.17) {\t};
}
\draw[<->] (.4,-.75)--(15.12,-.75);
\node[fill=white,inner sep=2pt] at (7.75,-.75) {Nominal CPI: 128 ms; 32 occasions};
\node[anchor=west,font=\scriptsize] at (.4,-1.22) {Each occasion starts every 4 ms (8 slots); pilots occupy 12 symbols in its first slot};
\end{scope}
\end{tikzpicture}
\caption{PRS evaluation waveform. (a) Disjoint TRP frequency bands each carry
three sequential spatial-code resources. (b) Twelve tones (one RB) show
comb-2 sampling. Symbols 1/3 and 2/4 sample equal frequencies, producing
early/late phase pairs after pilot removal. (c) Green bars mark PRS-bearing
slots. Doppler fitting spans 32 occasions within the 128-ms CPI; short
phase pairs resolve its speed ambiguity.}
\label{fig:prs_pattern}\end{figure*}
\subsection{PRS Waveform, Timing, and Evaluation Protocol}
Evaluation uses PRS from 3GPP New Radio (NR) Technical Specification
(TS)~38.211~\cite{ts38211,wei_prs}, with allocation/timing in
Fig.~\ref{fig:prs_pattern} and configuration in Table~\ref{tab:setup}.
A resource block (RB) contains 12 subcarriers; a resource element (RE)
is one subcarrier--symbol position. Each site sends three one-logical-port
resources in symbols 2--5, 6--9, and 10--13 of the first slot.
Comb 2 alternates even/odd tones: the first two symbols form the early view,
the last two the late view. Known pilots are removed before same-tone phase
comparison. Occasions start at 0,4,\ldots,124 ms;
Fig.~\ref{fig:prs_pattern} labels the matched-site bands and unused RB.
For the gates, $M$ counts group reports, $e$ is a standardized residual
coordinate, $Q$ its summed squared norm, and $\chi^2_d(q)$ is the chi-square
quantile at probability $q$ with $d$ degrees of freedom.

\begin{table}[t]\centering
\caption{Waveform benchmark configuration. Cooperation matches aggregate transmit resources.}
\label{tab:setup}\setlength{\tabcolsep}{3pt}
\begin{tabular}{@{}p{.38\columnwidth}p{.57\columnwidth}@{}}\toprule
Carrier / numerology & 4 GHz / 30 kHz subcarrier spacing, 273 carrier RBs\\
Full single-site PRS & 272 RBs, 52 dBm conducted power\\
Matched three-site PRS & 92/92/88 RBs; 47.229 dBm per site\\
PRS resources & Three one-logical-port resources; comb 2,
four symbols per resource\\
Array & $8\times8$ per receive polarization, $\pm45^\circ$ dual polarization;
128 physical ports; spacings $0.5\lambda/0.8\lambda$; $30^\circ$ uptilt\\
Transmit illumination & Three unit-norm spatial chirp masks;
equal collocated polarization weights synthesize local vertical polarization\\
Slow/fast timing & 32 occasions, 4 ms spacing; early/late pairs 71.354 $\mu$s apart\\
Nominal CPI / coverage & 128 ms / 5,120 target states per environment\\
Noise / residual interference & 290 K, 5 dB noise figure;
residual $-94$ dBm/100 MHz, fixed power spectral density\\
Numerical transforms & Range/Doppler 4,096/32; range search 420/800 m;
azimuth $\pm60^\circ$ locally, $2^\circ$ angle grids;
world elevation $[0,86]^\circ/[-86,86]^\circ$ (stochastic/city)\\
Pair / joint design gates & $\chi^2_3(0.9973)$;
$\max|e|\le4$, $Q\le\chi^2_{\max(4M-5,1)}(0.999)$\\\bottomrule
\end{tabular}\end{table}

For centered horizontal/vertical indices $y,z\in\{-3.5,-2.5,\ldots,3.5\}$
and resource $b\in\{1,2,3\}$, both TX polarizations use
$w_b(y,z)=e^{\mathrm j\psi_b(y,z)}/\sqrt{128}$, with phase $\psi_b$ in radians:
\begin{equation}
\begin{aligned}
\psi_b={}&-2\pi(0.5\xi_b y+0.8[0.173238+\xi_b]z)\\
&-\frac\pi7(0.995929s_{y,b}y^2+0.96s_{z,b}z^2).
\end{aligned}\label{eq:tx_masks}
\end{equation}
Quadratic signs $s_{y,b},s_{z,b}$ and linear gradient $\xi_b$ have triples
$(1,1,0)$, $(-1,1,0.12)$, and $(1,-1,-0.12)$.
The TR~38.901 element has 8 dBi peak gain, $65^\circ$ beamwidths, 30 dB cuts,
and polarization model~2~\cite{tr38901}. Panel rotation is
$\mathbf R_z(\theta_i)\mathbf R_y(-30^\circ)$: $\mathbf R_z,\mathbf R_y$
rotate about global $z,y$ axes, and
$\theta_i=(30,167.608,-107.608)^\circ$ gives site orientations.
Local $+x$ is boresight; row index increases downward.
Masks map logical ports across the array; three resources and two RX
polarizations yield six groups and 384 virtual channels.
The 4-ms and 71.354-$\mu$s phase intervals give radial-velocity ambiguity scales
$\pm4.684$ and $\pm262.59$ m/s, respectively; short pairs unfold slow-time aliases.
TX consistency is maximized within a spherical $4^\circ$ arrival cap.
Strong $(C_s,C_t,C_{\rm TX})$ thresholds are $(0.80,0.72,0.85)$;
eligibility uses $(0.50,0.45,0.40)$.
Strong/eligible half-CPI checks require minimum split-power ratios
$0.20/0.08$, maximum angular differences $8^\circ/16^\circ$,
and maximum unfolded-speed differences $1.25/2.5$ m/s.

Each environment has 1,024 occupied drops with five UAVs and 1,024 empty
drops. Horizontal positions are uniform in area in the $120^\circ$ sector of
Fig.~\ref{fig:deployment}; altitude is 25--300 m and horizontal speed
0--50 m/s at random heading. City draws require at least 1 m roof/terrain clearance.
Sites have 500 m horizontal separation and 25 m height above terrain.
These experiments set vertical velocity $v_z=0$ and fit five states with
a 60 m/s speed margin. Horizontal velocity requires rank-two noise-weighted
horizontal sightlines with condition number at most 100. Geometry is fixed within
the CPI; Doppler follows actual symbol times. Independent
128-occupied/128-empty stochastic drops check transferred settings;
city thresholds are unchanged.

Residual self-interference (SI) is independent post-cancellation Gaussian
noise with fixed power spectral density (PSD); intercell/inter-TRP
interference is absent. Ideal OFDM RE responses omit time-domain
intersymbol interference (ISI) and intercarrier interference (ICI).
Aggregate sensing occupancy is
$(272/273)\,12/(14\times8)\times(1/2)=5.3375\%$:
12 resource symbols per eight slots, comb 2, and 272 occupied carrier RBs,
giving 626,688 REs/CPI. Full-band fixed-site and aggregate matched
cooperation have equal transmitted budgets; cooperation has three receive
apertures. Policies receive identical reports within an allocation.
Observable selection ranks sites by strong count, summed bounded quality,
then identity, without fitting or truth. SC/PC/AC/SAR isolate admission
and lock/refit choices. Report-level controls replace Huber by ordinary
least squares (LS), covariance gates by Euclidean distance, or scale
the pair threshold.

Truth association maximizes cardinality within 30 m, then minimizes total
three-dimensional position distance. For drop $d$, let
$T_d,M_d,F_d$ count truths, matches, and unmatched outputs;
$N_o,N_e$ count occupied/empty drops, and $\mathcal D_+$ contains
occupied drops with outputs. Detection probability $P_D$ and
type-one/type-two false-alarm probabilities (FA1/FA2), denoted
$P_{\rm FA1},P_{\rm FA2}$, are
\begin{align}
P_D&=N_o^{-1}\sum_{d\text{ occupied}}M_d/T_d,\nonumber\\
P_{\rm FA1}&=N_e^{-1}\sum_{d\text{ empty}}\mathbf1\{M_d+F_d>0\},\nonumber\\
P_{\rm FA2}&=|\mathcal D_+|^{-1}
\sum_{d\in\mathcal D_+}F_d/(M_d+F_d).
\label{eq:metrics}
\end{align}
Here $\mathbf1$ is the indicator function and $|\mathcal D_+|$ the set
cardinality. Following
TR~38.765~\cite{tr38765}, FA2 averages per-drop false-output fractions.
The 90th-percentile error (P90) on own matches describes each detected
population; common-truth P90 compares identical target identities.
Controlled studies use mean-squared error (MSE) and Cram\'er--Rao bound (CRB).
Root-mean-square (RMS) generalized optimal subpattern assignment (GOSPA)
error~\cite{ospa,gospa} combines localization, misses and false outputs
with metric order $p=2$, cardinality penalty $\alpha=2$, and 30 m cutoff.
Pointwise 95\% intervals use 2,000 paired whole-drop bootstrap resamples.
Zero empty alarms in 1,024 trials gives exact one-sided 95\% upper
bound 0.2921\%.
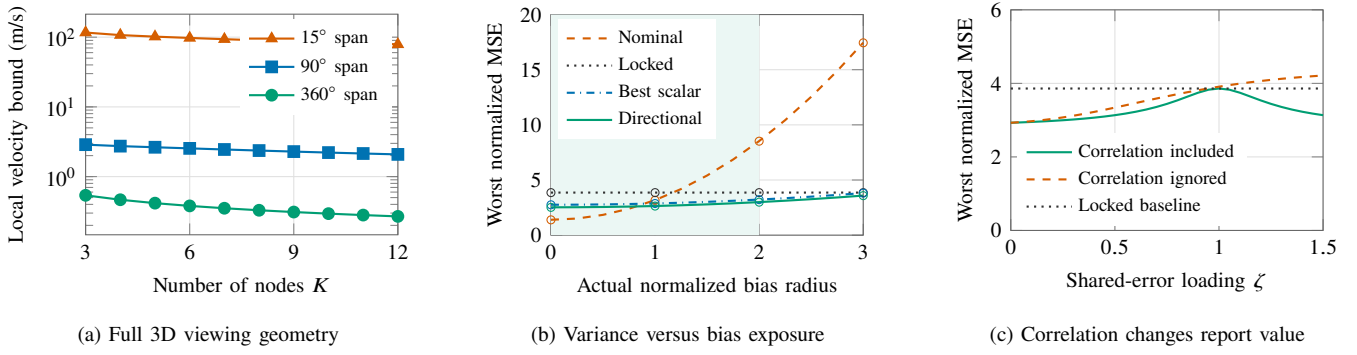
\begin{figure*}[t]\centering
\begin{minipage}[b]{.315\textwidth}\centering\begin{tikzpicture}\begin{axis}[paperaxis,width=\linewidth,height=4.5cm,xlabel={Number of nodes $K$},ylabel={Local velocity bound (m/s)},ymode=log,xmin=3,xmax=12,xtick={3,6,9,12},legend style={font=\scriptsize,at={(.98,.98)},anchor=north east}]
\addplot[paperred,mark=triangle*,thick] coordinates {(3,116.4614) (4,107.0922) (5,101.8116) (6,97.44031) (7,93.55678) (8,90.05297) (9,86.87465) (10,83.98088) (11,81.33686) (12,78.91234)};\addlegendentry{15$^\circ$ span}
\addplot[paperblue,mark=square*,thick] coordinates {(3,2.874446) (4,2.737199) (5,2.637078) (6,2.541108) (7,2.449821) (8,2.364456) (9,2.285372) (10,2.212389) (11,2.145084) (12,2.082952)};\addlegendentry{90$^\circ$ span}
\addplot[papergreen,mark=*,thick] coordinates {(3,0.539088) (4,0.4658475) (5,0.416448) (6,0.3801633) (7,0.3519628) (8,0.3292311) (9,0.310402) (10,0.2944732) (11,0.2807692) (12,0.268816)};\addlegendentry{360$^\circ$ span}
\end{axis}\end{tikzpicture}\par\smallskip{\footnotesize (a) Full 3D viewing geometry}\end{minipage}
\hfill
\begin{minipage}[b]{.315\textwidth}\centering\begin{tikzpicture}\begin{axis}[paperaxis,width=\linewidth,height=4.5cm,xlabel={Actual normalized bias radius},ylabel={Worst normalized MSE},xmin=0,xmax=3,ymin=0,ymax=20,legend style={font=\scriptsize,at={(.02,.98)},anchor=north west}]
\addplot[draw=none,fill=papergreen!8,forget plot] coordinates {(0,0)(2,0)(2,20)(0,20)} \closedcycle;
\addplot[paperred,dashed,thick] coordinates {(0,1.39493) (0.25,1.506427) (0.5,1.840919) (0.75,2.398405) (1,3.178886) (1.25,4.182362) (1.5,5.408832) (1.75,6.858296) (2,8.530756) (2.25,10.42621) (2.5,12.54466) (2.75,14.8861) (3,17.45054)};\addlegendentry{Nominal}
\addplot[only marks,mark=o,mark size=1.5pt,paperred,forget plot] coordinates {(0,1.388229) (1,3.170083) (2,8.519849) (3,17.43753)};
\addplot[paperblack,dotted,thick] coordinates {(0,3.86) (0.25,3.86) (0.5,3.86) (0.75,3.86) (1,3.86) (1.25,3.86) (1.5,3.86) (1.75,3.86) (2,3.86) (2.25,3.86) (2.5,3.86) (2.75,3.86) (3,3.86)};\addlegendentry{Locked}
\addplot[only marks,mark=o,mark size=1.5pt,paperblack,forget plot] coordinates {(0,3.856377) (1,3.856377) (2,3.856377) (3,3.856377)};
\addplot[paperblue,dashdotted,thick] coordinates {(0,2.75667) (0.25,2.76402) (0.5,2.786071) (0.75,2.822822) (1,2.874273) (1.25,2.940425) (1.5,3.021277) (1.75,3.11683) (2,3.227083) (2.25,3.352036) (2.5,3.49169) (2.75,3.646044) (3,3.815099)};\addlegendentry{Best scalar}
\addplot[only marks,mark=o,mark size=1.5pt,paperblue,forget plot] coordinates {(0,2.752069) (1,2.869129) (2,3.221396) (3,3.808869)};
\addplot[papergreen,solid,thick] coordinates {(0,2.52522) (0.25,2.532597) (0.5,2.554729) (0.75,2.591616) (1,2.643257) (1.25,2.709654) (1.5,2.790804) (1.75,2.88671) (2,2.99737) (2.25,3.122785) (2.5,3.262955) (2.75,3.417879) (3,3.587558)};\addlegendentry{Directional}
\addplot[only marks,mark=o,mark size=1.5pt,papergreen,forget plot] coordinates {(0,2.519726) (1,2.637666) (2,2.991682) (3,3.581773)};
\end{axis}\end{tikzpicture}\par\smallskip{\footnotesize (b) Variance versus bias exposure}\end{minipage}
\hfill
\begin{minipage}[b]{.315\textwidth}\centering\begin{tikzpicture}\begin{axis}[paperaxis,width=\linewidth,height=4.5cm,xlabel={Shared-error loading $\zeta$},ylabel={Worst normalized MSE},xmin=0,xmax=1.5,ymin=0,ymax=6,legend style={font=\scriptsize,at={(.02,.02)},anchor=south west}]
\addplot[papergreen,solid,thick] coordinates {(0,2.928721) (0.05,2.939255) (0.1,2.95114) (0.15,2.964599) (0.2,2.979899) (0.25,2.99737) (0.3,3.017415) (0.35,3.040523) (0.4,3.067295) (0.45,3.098466) (0.5,3.134939) (0.55,3.177813) (0.6,3.228407) (0.65,3.288253) (0.7,3.35899) (0.75,3.44202) (0.8,3.537575) (0.85,3.642555) (0.9,3.746577) (0.95,3.828213) (1,3.86) (1.05,3.828213) (1.1,3.746577) (1.15,3.642555) (1.2,3.537575) (1.25,3.44202) (1.3,3.35899) (1.35,3.288253) (1.4,3.228407) (1.45,3.177813) (1.5,3.134939)};\addlegendentry{Correlation included}
\addplot[paperred,dashed,thick] coordinates {(0,2.928721) (0.05,2.947406) (0.1,2.970905) (0.15,2.999704) (0.2,3.034062) (0.25,3.073986) (0.3,3.119228) (0.35,3.169308) (0.4,3.223555) (0.45,3.281161) (0.5,3.341238) (0.55,3.402874) (0.6,3.465178) (0.65,3.527326) (0.7,3.588578) (0.75,3.648305) (0.8,3.705986) (0.85,3.761216) (0.9,3.813691) (0.95,3.863206) (1,3.909638) (1.05,3.952935) (1.1,3.993105) (1.15,4.030202) (1.2,4.064316) (1.25,4.095565) (1.3,4.124084) (1.35,4.150019) (1.4,4.173524) (1.45,4.194755) (1.5,4.213868)};\addlegendentry{Correlation ignored}
\addplot[paperblack,dotted,thick] coordinates {(0,3.86)(1.5,3.86)};\addlegendentry{Locked baseline}
\end{axis}\end{tikzpicture}\par\smallskip{\footnotesize (c) Correlation changes report value}\end{minipage}
\caption{General-model controls. (a) Simultaneous monostatic nodes at 500 m horizontal range observe $\mathbf v=(25,-15,8)$ m/s. Range, angle, and radial-speed error standard deviations are $\sigma_r=1$ m, $\sigma_{\rm angle}=0.5^\circ$, and $\sigma_\nu=0.15$ m/s. Two views have deficient 3D velocity rank. (b) The six-state/eight-coordinate linearized monostatic fixture compares exact worst risks and 200,000-sample Monte Carlo marks, using each method's worst bias direction. Gains use design radius $\rho=2$; shading marks actual bias radii at most two. (c) Shared-error loading $\zeta$ in $\boldsymbol\epsilon=\zeta\mathbf H\mathbf e_S+\boldsymbol\epsilon_0$ sets report/baseline correlation; $\boldsymbol\epsilon_0$ is independent noise. At $\zeta=1$ the innovation carries no baseline-error information.}
\label{fig:controlled}\end{figure*}

\begin{figure*}[t]\centering
\begin{tikzpicture}
\begin{axis}[name=benefit,paperaxis,width=.48\textwidth,height=3.6cm,ybar,bar width=6pt,ylabel={Share of false-output\\reduction (\%)},ylabel style={align=center},xtick={0,...,6},xticklabels={Spatial,Temporal,TX signature,Delayed,Split angle,Split velocity,Split power},x tick label style={rotate=35,anchor=east,font=\scriptsize},ymin=0,ymax=55]
\addplot[fill=paperblue,draw=paperblue] coordinates {(0,22.56748) (1,43.50058) (2,16.93151) (3,11.97042) (4,4.31154) (5,0.6055553) (6,0.1129033)};
\addplot[fill=papergreen,draw=papergreen] coordinates {(0,31.68666) (1,28.31382) (2,23.99719) (3,9.806008) (4,6.064934) (5,0.0909264) (6,0.04046463)};
\end{axis}\node[overlay,anchor=south,font=\footnotesize] at ([xshift=8.0cm,yshift=6pt]benefit.north west) {\textcolor{paperblue}{\rule{7pt}{7pt}}\enspace 3GPP stochastic environment\hspace{1cm}\textcolor{papergreen}{\rule{7pt}{7pt}}\enspace Mapped city with ray tracing};\path (benefit.north west) ++(0,20pt);\node[font=\footnotesize] at (current bounding box.south) [below=3pt] {(a) Rejection benefit};\end{tikzpicture}\hfill
\begin{tikzpicture}
\begin{axis}[paperaxis,width=.48\textwidth,height=3.6cm,ybar,bar width=6pt,ylabel={Miss increase attributed\\to check (pp)},ylabel style={align=center},xtick={0,...,6},xticklabels={Spatial,Temporal,TX signature,Delayed,Split angle,Split velocity,Split power},x tick label style={rotate=35,anchor=east,font=\scriptsize},ymin=0,ymax=3]
\addplot[fill=paperblue,draw=paperblue] coordinates {(0,0.3336589) (1,2.413737) (2,0.1383464) (3,0.1009115) (4,0.1009115) (5,0.01464844) (6,0.003255208)};
\addplot[fill=papergreen,draw=papergreen] coordinates {(0,1.585286) (1,0.8186849) (2,0.7747396) (3,0.1188151) (4,0.3059896) (5,-0) (6,0.009765625)};
\end{axis}\node[font=\footnotesize] at (current bounding box.south) [below=3pt] {(b) Sensitivity cost};\end{tikzpicture}
\caption{Exact seven-check Shapley attribution from all $2^7=128$ coalitions, full-band single TRP, 1,024 occupied drops per environment. Benefit is normalized by total removed false outputs (30.368/drop for 3GPP, 6.838/drop for San Francisco). Sensitivity cost sums to 3.105 and 3.613 percentage points, respectively. The coalition-averaged marginal credits describe these populations and include interactions; they are not intrinsic physical rankings.}
\label{fig:ablation}\end{figure*}
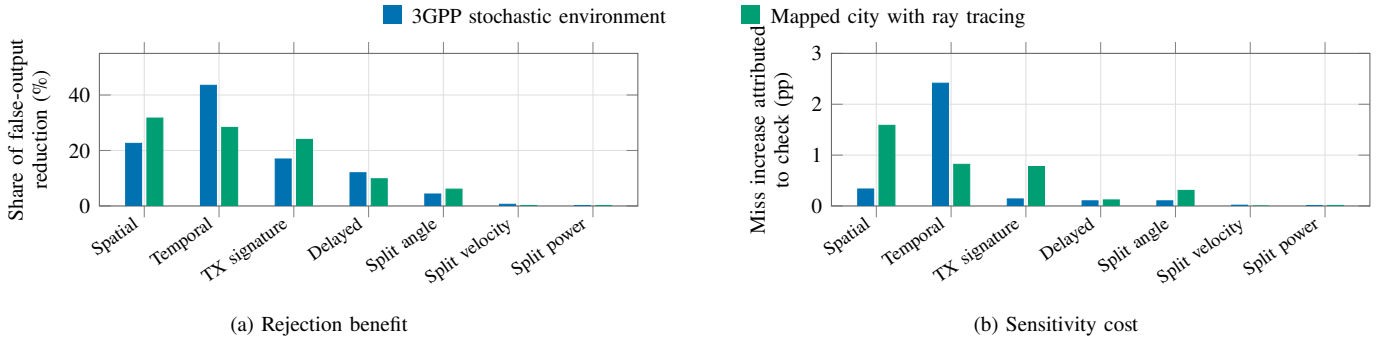

\begin{table*}[t]\centering
\caption{Primary matched-budget results. Five targets per occupied drop; 5,120 truths per environment. FA2 is a mean occupied-drop fraction. Euclidean 30 m is the exploratory strong association comparator; refinement preserves these counts.}
\label{tab:primary_results}
\begin{tabular}{lrrrrrrrr}\toprule
&\multicolumn{4}{c}{3GPP stochastic environment}&\multicolumn{4}{c}{Mapped city with ray tracing}\\
Policy & $P_D$ (\%) & FA2 (\%) & Matched & Unmatched & $P_D$ (\%) & FA2 (\%) & Matched & Unmatched\\\midrule
Observable site selection & 97.93 & 0.614 & 5,014 & 32 & 88.75 & 11.056 & 4,544 & 595\\
SC & 96.21 & 0.036 & 4,926 & 2 & 89.16 & 0.120 & 4,565 & 7\\
PC & 96.80 & 1.285 & 4,956 & 77 & 92.71 & 0.568 & 4,747 & 32\\
AC & 96.80 & 0.088 & 4,956 & 5 & 92.46 & 0.396 & 4,734 & 22\\
SAR & 96.80 & 0.088 & 4,956 & 5 & 92.46 & 0.396 & 4,734 & 22\\
SAR, Euclidean 30 m & 96.78 & 0.036 & 4,955 & 2 & 92.50 & 0.091 & 4,736 & 5\\
\bottomrule\end{tabular}\end{table*}

\begin{figure*}[t]\centering
\begin{minipage}[b]{.48\textwidth}\centering\begin{tikzpicture}\begin{axis}[paperaxis,width=\linewidth,height=4.0cm,xmode=log,xlabel={Occupied-drop false-output fraction (\%)},ylabel={Detection (\%)},xmin=.015,xmax=2.1,ymin=95.2,ymax=98.8,legend style={at={(.02,.02)},anchor=south west,font=\scriptsize,legend columns=2}]
\addplot[only marks,mark=square*,mark size=3pt,paperblack,error bars/.cd,y dir=both,y explicit] coordinates {(0.61360677,97.929688) +=(0,0.390625) -=(0,0.41015625)};\addlegendentry{Site selection}
\addplot[only marks,mark=*,mark size=3pt,paperblue,error bars/.cd,y dir=both,y explicit] coordinates {(0.035807292,96.210938) +=(0,0.5078125) -=(0,0.52734375)};\addlegendentry{SC}
\addplot[only marks,mark=triangle*,mark size=3pt,paperred,error bars/.cd,y dir=both,y explicit] coordinates {(1.2846447,96.796875) +=(0,0.46875) -=(0,0.5078125)};\addlegendentry{PC}
\addplot[only marks,mark=diamond*,mark size=3pt,papergreen,error bars/.cd,y dir=both,y explicit] coordinates {(0.087890625,96.796875) +=(0,0.46875) -=(0,0.5078125)};\addlegendentry{AC / SAR}
\end{axis}\end{tikzpicture}\par\smallskip{\footnotesize(a) 3GPP stochastic environment}\end{minipage}
\hfill
\begin{minipage}[b]{.48\textwidth}\centering\begin{tikzpicture}\begin{axis}[paperaxis,width=\linewidth,height=4.0cm,xmode=log,xlabel={Occupied-drop false-output fraction (\%)},ylabel={Detection (\%)},xmin=.035,xmax=18,ymin=86.8,ymax=94.5,legend style={at={(.02,.02)},anchor=south west,font=\scriptsize,legend columns=2}]
\addplot[only marks,mark=square*,mark size=3pt,paperblack,error bars/.cd,y dir=both,y explicit] coordinates {(11.055966,88.75) +=(0,0.95703125) -=(0,0.95751953)};\addlegendentry{Site selection}
\addplot[only marks,mark=*,mark size=3pt,paperblue,error bars/.cd,y dir=both,y explicit] coordinates {(0.12044271,89.160156) +=(0,0.83984375) -=(0,0.87890625)};\addlegendentry{SC}
\addplot[only marks,mark=triangle*,mark size=3pt,paperred,error bars/.cd,y dir=both,y explicit] coordinates {(0.56803385,92.714844) +=(0,0.72265625) -=(0,0.68359375)};\addlegendentry{PC}
\addplot[only marks,mark=diamond*,mark size=3pt,papergreen,error bars/.cd,y dir=both,y explicit] coordinates {(0.39550781,92.460938) +=(0,0.68359375) -=(0,0.703125)};\addlegendentry{AC / SAR}
\end{axis}\end{tikzpicture}\par\smallskip{\footnotesize(b) Mapped city with ray tracing}\end{minipage}
\caption{The same matched-budget site observations support selection-only and fusion policies. Vertical bars are 95\% whole-drop detection intervals; the horizontal axis is logarithmic. AC and SAR coincide in detection/false-output scores while differing in fitted states. These are fixed operating points, not an optimized receiver operating frontier.}
\label{fig:selection_tradeoff}\end{figure*}
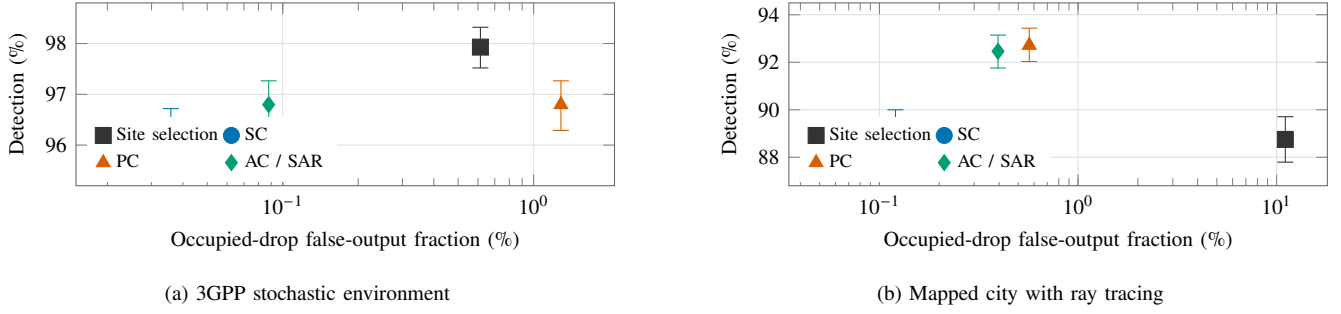

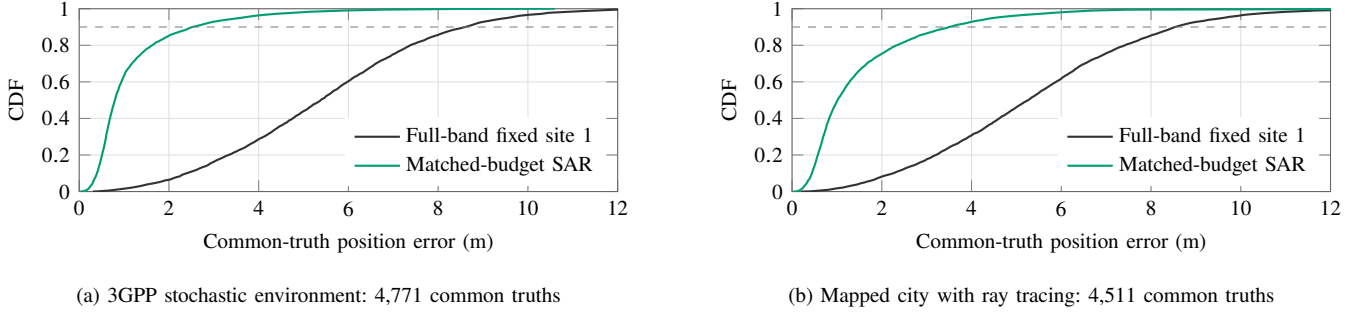
\begin{figure*}[t]\centering
\begin{minipage}[b]{.48\textwidth}\centering\begin{tikzpicture}\begin{axis}[paperaxis,width=\linewidth,height=4.0cm,xlabel={Common-truth position error (m)},ylabel={CDF},xmin=0,xmax=12,ymin=0,ymax=1,legend style={font=\footnotesize,at={(.98,.03)},anchor=south east}]
\addplot[paperblack,thick] coordinates {(0.30734524,0) (0.40810628,0.001) (0.63943059,0.005) (0.65840103,0.0055555556) (0.81849246,0.01) (0.85193661,0.011111111) (1.0358493,0.016666667) (1.1866692,0.022222222) (1.3148655,0.027777778) (1.4233933,0.033333333) (1.5104712,0.038888889) (1.642838,0.044444444) (1.7551735,0.05) (1.8469005,0.055555556) (1.9200395,0.061111111) (2.0411027,0.066666667) (2.0895038,0.072222222) (2.1553002,0.077777778) (2.2442061,0.083333333) (2.2760559,0.088888889) (2.3229176,0.094444444) (2.3944806,0.1) (2.4590039,0.10555556) (2.5266623,0.11111111) (2.5928007,0.11666667) (2.6461727,0.12222222) (2.6937249,0.12777778) (2.7498698,0.13333333) (2.8168592,0.13888889) (2.867023,0.14444444) (2.9033951,0.15) (2.9359521,0.15555556) (2.9872855,0.16111111) (3.0288185,0.16666667) (3.0723592,0.17222222) (3.126239,0.17777778) (3.1811573,0.18333333) (3.223304,0.18888889) (3.2700664,0.19444444) (3.3279939,0.2) (3.3915037,0.20555556) (3.4243841,0.21111111) (3.4775605,0.21666667) (3.5274839,0.22222222) (3.5661167,0.22777778) (3.6245528,0.23333333) (3.6718707,0.23888889) (3.7168289,0.24444444) (3.7490247,0.25) (3.7858817,0.25555556) (3.8251257,0.26111111) (3.8688036,0.26666667) (3.9106553,0.27222222) (3.9463738,0.27777778) (3.9828499,0.28333333) (4.0181728,0.28888889) (4.0701421,0.29444444) (4.1122822,0.3) (4.1573555,0.30555556) (4.1938536,0.31111111) (4.2243501,0.31666667) (4.2664862,0.32222222) (4.2978131,0.32777778) (4.328058,0.33333333) (4.3583551,0.33888889) (4.3948607,0.34444444) (4.4270696,0.35) (4.4576916,0.35555556) (4.5019694,0.36111111) (4.5462953,0.36666667) (4.5749392,0.37222222) (4.6267922,0.37777778) (4.6523613,0.38333333) (4.677573,0.38888889) (4.7281662,0.39444444) (4.75664,0.4) (4.7928703,0.40555556) (4.816854,0.41111111) (4.8559211,0.41666667) (4.8960494,0.42222222) (4.9281569,0.42777778) (4.9576943,0.43333333) (4.9934033,0.43888889) (5.0395977,0.44444444) (5.0776735,0.45) (5.1084605,0.45555556) (5.1316732,0.46111111) (5.1735067,0.46666667) (5.1981691,0.47222222) (5.2293707,0.47777778) (5.2835249,0.48333333) (5.3180339,0.48888889) (5.3575878,0.49444444) (5.3844265,0.5) (5.4094453,0.50555556) (5.4586108,0.51111111) (5.4836738,0.51666667) (5.5049854,0.52222222) (5.543227,0.52777778) (5.5798871,0.53333333) (5.6136278,0.53888889) (5.6377115,0.54444444) (5.6674798,0.55) (5.7076413,0.55555556) (5.7533144,0.56111111) (5.7777997,0.56666667) (5.804825,0.57222222) (5.833882,0.57777778) (5.8673018,0.58333333) (5.9012835,0.58888889) (5.9422297,0.59444444) (5.976563,0.6) (6.0134436,0.60555556) (6.0473522,0.61111111) (6.0735798,0.61666667) (6.1061059,0.62222222) (6.141308,0.62777778) (6.1924306,0.63333333) (6.2227641,0.63888889) (6.2622639,0.64444444) (6.3132251,0.65) (6.3464949,0.65555556) (6.3796578,0.66111111) (6.4113727,0.66666667) (6.4440511,0.67222222) (6.4772282,0.67777778) (6.5141808,0.68333333) (6.5510862,0.68888889) (6.5761606,0.69444444) (6.6113639,0.7) (6.6481569,0.70555556) (6.6881599,0.71111111) (6.7292344,0.71666667) (6.782333,0.72222222) (6.8315926,0.72777778) (6.8722104,0.73333333) (6.9257609,0.73888889) (6.9560576,0.74444444) (6.9898489,0.75) (7.0423234,0.75555556) (7.0788354,0.76111111) (7.1216357,0.76666667) (7.1654183,0.77222222) (7.2114844,0.77777778) (7.2527266,0.78333333) (7.2936235,0.78888889) (7.3375795,0.79444444) (7.3752593,0.8) (7.4318625,0.80555556) (7.4991001,0.81111111) (7.5487863,0.81666667) (7.6204778,0.82222222) (7.6898704,0.82777778) (7.7324369,0.83333333) (7.7820767,0.83888889) (7.8438364,0.84444444) (7.9071182,0.85) (7.9843572,0.85555556) (8.0384061,0.86111111) (8.1022,0.86666667) (8.1670115,0.87222222) (8.2393715,0.87777778) (8.3152516,0.88333333) (8.4130673,0.88888889) (8.5148861,0.89444444) (8.606955,0.9) (8.683697,0.90555556) (8.7561159,0.91111111) (8.8250588,0.91666667) (8.8908124,0.92222222) (8.9880479,0.92777778) (9.1213766,0.93333333) (9.2463151,0.93888889) (9.3796848,0.94444444) (9.5520789,0.95) (9.5520789,0.95) (9.6822468,0.95555556) (9.7995224,0.96111111) (10.012361,0.96666667) (10.306026,0.97222222) (10.457673,0.97777778) (10.899466,0.98333333) (11.372764,0.98888889) (11.487949,0.99) (11.91069,0.99444444) (11.94807,0.995) (13.611839,0.999) (17.391694,1)};\addlegendentry{Full-band fixed site 1}
\addplot[papergreen,thick] coordinates {(0.024443695,0) (0.080133047,0.001) (0.13488834,0.005) (0.14275488,0.0055555556) (0.17834564,0.01) (0.1838281,0.011111111) (0.21209064,0.016666667) (0.22759056,0.022222222) (0.2448892,0.027777778) (0.26481054,0.033333333) (0.28312179,0.038888889) (0.29537864,0.044444444) (0.30475365,0.05) (0.31453687,0.055555556) (0.32492326,0.061111111) (0.33371403,0.066666667) (0.34741141,0.072222222) (0.36088856,0.077777778) (0.36827582,0.083333333) (0.37902764,0.088888889) (0.38883773,0.094444444) (0.39432184,0.1) (0.40385252,0.10555556) (0.41166835,0.11111111) (0.41938149,0.11666667) (0.42561117,0.12222222) (0.43131737,0.12777778) (0.43749706,0.13333333) (0.44375119,0.13888889) (0.4506946,0.14444444) (0.45750031,0.15) (0.4642527,0.15555556) (0.46957625,0.16111111) (0.47585462,0.16666667) (0.47952271,0.17222222) (0.48591062,0.17777778) (0.49211296,0.18333333) (0.49753443,0.18888889) (0.50366136,0.19444444) (0.50907831,0.2) (0.513466,0.20555556) (0.51808906,0.21111111) (0.52495897,0.21666667) (0.53019036,0.22222222) (0.53657557,0.22777778) (0.5432735,0.23333333) (0.54824999,0.23888889) (0.55221054,0.24444444) (0.55838666,0.25) (0.56410701,0.25555556) (0.56840396,0.26111111) (0.57338053,0.26666667) (0.57806063,0.27222222) (0.58187303,0.27777778) (0.5871737,0.28333333) (0.5895303,0.28888889) (0.59366822,0.29444444) (0.59773814,0.3) (0.60079869,0.30555556) (0.60403674,0.31111111) (0.60886305,0.31666667) (0.61489055,0.32222222) (0.62042787,0.32777778) (0.62491666,0.33333333) (0.63139408,0.33888889) (0.63562586,0.34444444) (0.64059226,0.35) (0.64459914,0.35555556) (0.64916067,0.36111111) (0.65296398,0.36666667) (0.66004916,0.37222222) (0.66507733,0.37777778) (0.66996295,0.38333333) (0.67498005,0.38888889) (0.67786285,0.39444444) (0.68438691,0.4) (0.69202682,0.40555556) (0.69626903,0.41111111) (0.70131957,0.41666667) (0.7085058,0.42222222) (0.71272437,0.42777778) (0.72017867,0.43333333) (0.72478595,0.43888889) (0.73051582,0.44444444) (0.73594079,0.45) (0.74202969,0.45555556) (0.74856468,0.46111111) (0.75453568,0.46666667) (0.76162798,0.47222222) (0.77012744,0.47777778) (0.77521345,0.48333333) (0.78093376,0.48888889) (0.7874036,0.49444444) (0.79491948,0.5) (0.80136781,0.50555556) (0.80854328,0.51111111) (0.81453699,0.51666667) (0.82004753,0.52222222) (0.82682437,0.52777778) (0.83328077,0.53333333) (0.84191571,0.53888889) (0.85078418,0.54444444) (0.8605704,0.55) (0.86901835,0.55555556) (0.87882468,0.56111111) (0.88375482,0.56666667) (0.89024617,0.57222222) (0.89963918,0.57777778) (0.90818528,0.58333333) (0.91877816,0.58888889) (0.93206614,0.59444444) (0.93863782,0.6) (0.94953389,0.60555556) (0.95717495,0.61111111) (0.96821832,0.61666667) (0.97811487,0.62222222) (0.98705136,0.62777778) (1.0014388,0.63333333) (1.0088171,0.63888889) (1.0172157,0.64444444) (1.025566,0.65) (1.0379209,0.65555556) (1.0553058,0.66111111) (1.07073,0.66666667) (1.0868872,0.67222222) (1.1066614,0.67777778) (1.1247369,0.68333333) (1.1374343,0.68888889) (1.1559473,0.69444444) (1.1728909,0.7) (1.1916995,0.70555556) (1.2112735,0.71111111) (1.2352113,0.71666667) (1.2588862,0.72222222) (1.2716822,0.72777778) (1.2951971,0.73333333) (1.319081,0.73888889) (1.3481456,0.74444444) (1.3667231,0.75) (1.3891482,0.75555556) (1.415093,0.76111111) (1.4386982,0.76666667) (1.4678659,0.77222222) (1.4890863,0.77777778) (1.5239824,0.78333333) (1.5626075,0.78888889) (1.5905897,0.79444444) (1.6259578,0.8) (1.6612243,0.80555556) (1.6887652,0.81111111) (1.7196614,0.81666667) (1.7725142,0.82222222) (1.8134501,0.82777778) (1.8605034,0.83333333) (1.8985281,0.83888889) (1.9391474,0.84444444) (1.9797322,0.85) (2.0258187,0.85555556) (2.0778689,0.86111111) (2.145399,0.86666667) (2.2154891,0.87222222) (2.277933,0.87777778) (2.3580165,0.88333333) (2.412002,0.88888889) (2.4620711,0.89444444) (2.521738,0.9) (2.6163846,0.90555556) (2.715154,0.91111111) (2.7878264,0.91666667) (2.8772457,0.92222222) (2.9673965,0.92777778) (3.1100748,0.93333333) (3.2465374,0.93888889) (3.4279146,0.94444444) (3.5961142,0.95) (3.5961142,0.95) (3.7501479,0.95555556) (3.9108986,0.96111111) (4.1192609,0.96666667) (4.3961851,0.97222222) (4.76245,0.97777778) (5.1252283,0.98333333) (5.75496,0.98888889) (5.864069,0.99) (6.671653,0.99444444) (6.7730185,0.995) (8.5061691,0.999) (10.597571,1)};\addlegendentry{Matched-budget SAR}
\addplot[gray,dashed] coordinates {(0,.9)(12,.9)};\end{axis}\end{tikzpicture}\par\smallskip{\footnotesize (a) 3GPP stochastic environment: 4,771 common truths}\end{minipage}
\hfill
\begin{minipage}[b]{.48\textwidth}\centering\begin{tikzpicture}\begin{axis}[paperaxis,width=\linewidth,height=4.0cm,xlabel={Common-truth position error (m)},ylabel={CDF},xmin=0,xmax=12,ymin=0,ymax=1,legend style={font=\footnotesize,at={(.98,.03)},anchor=south east}]
\addplot[paperblack,thick] coordinates {(0.20572484,0) (0.3965184,0.001) (0.60367395,0.005) (0.62533063,0.0055555556) (0.82292329,0.01) (0.84812214,0.011111111) (0.97376422,0.016666667) (1.1412067,0.022222222) (1.2368613,0.027777778) (1.3449593,0.033333333) (1.4560997,0.038888889) (1.5202056,0.044444444) (1.6177637,0.05) (1.6892617,0.055555556) (1.7815428,0.061111111) (1.8476776,0.066666667) (1.8970001,0.072222222) (1.9583451,0.077777778) (2.0178118,0.083333333) (2.0944087,0.088888889) (2.1827345,0.094444444) (2.2541073,0.1) (2.3111491,0.10555556) (2.3628616,0.11111111) (2.4310483,0.11666667) (2.4856737,0.12222222) (2.5326461,0.12777778) (2.5900051,0.13333333) (2.6556308,0.13888889) (2.7259718,0.14444444) (2.7784014,0.15) (2.8218746,0.15555556) (2.8693628,0.16111111) (2.9329771,0.16666667) (2.976854,0.17222222) (3.011859,0.17777778) (3.0727288,0.18333333) (3.1138469,0.18888889) (3.1520982,0.19444444) (3.2068439,0.2) (3.2551614,0.20555556) (3.2869761,0.21111111) (3.3190955,0.21666667) (3.3584503,0.22222222) (3.4113046,0.22777778) (3.4457828,0.23333333) (3.4868381,0.23888889) (3.5239489,0.24444444) (3.5657442,0.25) (3.6157258,0.25555556) (3.6741302,0.26111111) (3.7153767,0.26666667) (3.7456344,0.27222222) (3.7855089,0.27777778) (3.8246486,0.28333333) (3.878594,0.28888889) (3.9106966,0.29444444) (3.9423426,0.3) (3.9901026,0.30555556) (4.0214638,0.31111111) (4.0545775,0.31666667) (4.1221544,0.32222222) (4.1626985,0.32777778) (4.2051254,0.33333333) (4.2421457,0.33888889) (4.2715804,0.34444444) (4.3057911,0.35) (4.336707,0.35555556) (4.3713179,0.36111111) (4.4016385,0.36666667) (4.4536792,0.37222222) (4.4765154,0.37777778) (4.5159065,0.38333333) (4.5518975,0.38888889) (4.5841968,0.39444444) (4.6174332,0.4) (4.647149,0.40555556) (4.6778652,0.41111111) (4.7304215,0.41666667) (4.7598592,0.42222222) (4.7958397,0.42777778) (4.8300239,0.43333333) (4.8602235,0.43888889) (4.8924124,0.44444444) (4.9379447,0.45) (4.9654861,0.45555556) (5.0013761,0.46111111) (5.0313259,0.46666667) (5.0687789,0.47222222) (5.1067581,0.47777778) (5.1432355,0.48333333) (5.1755049,0.48888889) (5.2196266,0.49444444) (5.2605424,0.5) (5.2929459,0.50555556) (5.3247266,0.51111111) (5.3515099,0.51666667) (5.3982715,0.52222222) (5.4264214,0.52777778) (5.4649989,0.53333333) (5.4950375,0.53888889) (5.522971,0.54444444) (5.5542915,0.55) (5.5928663,0.55555556) (5.6362032,0.56111111) (5.6690279,0.56666667) (5.7039241,0.57222222) (5.7393389,0.57777778) (5.7689281,0.58333333) (5.8066108,0.58888889) (5.8460565,0.59444444) (5.8824246,0.6) (5.9205097,0.60555556) (5.9613281,0.61111111) (5.9966938,0.61666667) (6.0185292,0.62222222) (6.0545589,0.62777778) (6.0845074,0.63333333) (6.131573,0.63888889) (6.1605529,0.64444444) (6.2013763,0.65) (6.2387512,0.65555556) (6.2739375,0.66111111) (6.2965566,0.66666667) (6.3308501,0.67222222) (6.3647687,0.67777778) (6.393173,0.68333333) (6.4304039,0.68888889) (6.4702668,0.69444444) (6.5234759,0.7) (6.5701659,0.70555556) (6.6102114,0.71111111) (6.6477942,0.71666667) (6.6967982,0.72222222) (6.7483597,0.72777778) (6.8157147,0.73333333) (6.8610676,0.73888889) (6.9053778,0.74444444) (6.9511119,0.75) (6.9952011,0.75555556) (7.0432035,0.76111111) (7.1029623,0.76666667) (7.1484286,0.77222222) (7.1940002,0.77777778) (7.2374933,0.78333333) (7.2900089,0.78888889) (7.3599466,0.79444444) (7.4201673,0.8) (7.4625358,0.80555556) (7.5442968,0.81111111) (7.6019268,0.81666667) (7.6687821,0.82222222) (7.725882,0.82777778) (7.7648631,0.83333333) (7.8299518,0.83888889) (7.9048831,0.84444444) (7.9606626,0.85) (8.0191882,0.85555556) (8.0963936,0.86111111) (8.1630829,0.86666667) (8.225605,0.87222222) (8.28644,0.87777778) (8.3610527,0.88333333) (8.4160121,0.88888889) (8.4758098,0.89444444) (8.5359967,0.9) (8.6146352,0.90555556) (8.6955791,0.91111111) (8.7860977,0.91666667) (8.8760576,0.92222222) (8.9976627,0.92777778) (9.1683796,0.93333333) (9.2741285,0.93888889) (9.452281,0.94444444) (9.6181954,0.95) (9.6181954,0.95) (9.7461742,0.95555556) (9.9386828,0.96111111) (10.077869,0.96666667) (10.328256,0.97222222) (10.632602,0.97777778) (11.013453,0.98333333) (11.56113,0.98888889) (11.624242,0.99) (12.437838,0.99444444) (12.611472,0.995) (15.796551,0.999) (24.479065,1)};\addlegendentry{Full-band fixed site 1}
\addplot[papergreen,thick] coordinates {(0.026578655,0) (0.10529007,0.001) (0.15520458,0.005) (0.15890937,0.0055555556) (0.19370543,0.01) (0.19993172,0.011111111) (0.23038525,0.016666667) (0.24698134,0.022222222) (0.2647642,0.027777778) (0.28654538,0.033333333) (0.30344663,0.038888889) (0.32005988,0.044444444) (0.33755191,0.05) (0.34854242,0.055555556) (0.35775281,0.061111111) (0.37929493,0.066666667) (0.39474964,0.072222222) (0.40292954,0.077777778) (0.4157799,0.083333333) (0.42442644,0.088888889) (0.43459336,0.094444444) (0.44218337,0.1) (0.44902147,0.10555556) (0.45800612,0.11111111) (0.46354388,0.11666667) (0.47139559,0.12222222) (0.48173923,0.12777778) (0.48974701,0.13333333) (0.49809387,0.13888889) (0.50512457,0.14444444) (0.51149171,0.15) (0.52035072,0.15555556) (0.52692466,0.16111111) (0.53315077,0.16666667) (0.53990216,0.17222222) (0.5459667,0.17777778) (0.55229969,0.18333333) (0.55859339,0.18888889) (0.56724196,0.19444444) (0.57442723,0.2) (0.57811295,0.20555556) (0.58885035,0.21111111) (0.59263272,0.21666667) (0.60087716,0.22222222) (0.60994818,0.22777778) (0.61685458,0.23333333) (0.62350316,0.23888889) (0.629401,0.24444444) (0.63946982,0.25) (0.64547472,0.25555556) (0.65202486,0.26111111) (0.65784961,0.26666667) (0.66424793,0.27222222) (0.67318874,0.27777778) (0.67681164,0.28333333) (0.68183459,0.28888889) (0.68659318,0.29444444) (0.6945158,0.3) (0.70254865,0.30555556) (0.70896702,0.31111111) (0.71938458,0.31666667) (0.72595487,0.32222222) (0.73279722,0.32777778) (0.73947043,0.33333333) (0.74670236,0.33888889) (0.75276265,0.34444444) (0.75878212,0.35) (0.76516713,0.35555556) (0.77152459,0.36111111) (0.77660811,0.36666667) (0.78540448,0.37222222) (0.79540648,0.37777778) (0.8047922,0.38333333) (0.8138905,0.38888889) (0.81963149,0.39444444) (0.82773053,0.4) (0.83597516,0.40555556) (0.84596362,0.41111111) (0.85430563,0.41666667) (0.86564987,0.42222222) (0.87261464,0.42777778) (0.88386271,0.43333333) (0.8916062,0.43888889) (0.90333154,0.44444444) (0.91191278,0.45) (0.92021707,0.45555556) (0.93302185,0.46111111) (0.9460074,0.46666667) (0.95713929,0.47222222) (0.97080361,0.47777778) (0.97934175,0.48333333) (0.99015531,0.48888889) (1.0030439,0.49444444) (1.0148818,0.5) (1.0259892,0.50555556) (1.0377929,0.51111111) (1.0546234,0.51666667) (1.0736655,0.52222222) (1.0857978,0.52777778) (1.0987232,0.53333333) (1.114216,0.53888889) (1.128834,0.54444444) (1.1464093,0.55) (1.1586906,0.55555556) (1.1701934,0.56111111) (1.1838613,0.56666667) (1.1981599,0.57222222) (1.2139947,0.57777778) (1.2294663,0.58333333) (1.2457752,0.58888889) (1.2655725,0.59444444) (1.2797293,0.6) (1.2919623,0.60555556) (1.3150298,0.61111111) (1.3322578,0.61666667) (1.3532779,0.62222222) (1.3750011,0.62777778) (1.3938818,0.63333333) (1.4069148,0.63888889) (1.4298401,0.64444444) (1.4492015,0.65) (1.4685066,0.65555556) (1.4926604,0.66111111) (1.5188274,0.66666667) (1.5406015,0.67222222) (1.561423,0.67777778) (1.5946296,0.68333333) (1.6234796,0.68888889) (1.6452987,0.69444444) (1.6703731,0.7) (1.6902574,0.70555556) (1.7260992,0.71111111) (1.7706874,0.71666667) (1.7982205,0.72222222) (1.8243803,0.72777778) (1.8604197,0.73333333) (1.8952867,0.73888889) (1.9377668,0.74444444) (1.9743437,0.75) (2.014986,0.75555556) (2.0490571,0.76111111) (2.0880506,0.76666667) (2.1248085,0.77222222) (2.1585763,0.77777778) (2.2019218,0.78333333) (2.2331265,0.78888889) (2.280682,0.79444444) (2.3312971,0.8) (2.385294,0.80555556) (2.4390361,0.81111111) (2.4895234,0.81666667) (2.5440985,0.82222222) (2.5933696,0.82777778) (2.6446368,0.83333333) (2.7029967,0.83888889) (2.7585172,0.84444444) (2.8022365,0.85) (2.8569835,0.85555556) (2.9457487,0.86111111) (3.0184445,0.86666667) (3.0947438,0.87222222) (3.172104,0.87777778) (3.2398838,0.88333333) (3.3397761,0.88888889) (3.4395272,0.89444444) (3.4968348,0.9) (3.5962917,0.90555556) (3.6617028,0.91111111) (3.7866421,0.91666667) (3.897718,0.92222222) (3.9862623,0.92777778) (4.095963,0.93333333) (4.2465101,0.93888889) (4.3512713,0.94444444) (4.4966067,0.95) (4.4966067,0.95) (4.6918728,0.95555556) (4.9340098,0.96111111) (5.1930558,0.96666667) (5.5320345,0.97222222) (5.8228556,0.97777778) (6.1776378,0.98333333) (6.6067013,0.98888889) (6.6916166,0.99) (7.4516736,0.99444444) (7.6396845,0.995) (13.466863,0.999) (27.437303,1)};\addlegendentry{Matched-budget SAR}
\addplot[gray,dashed] coordinates {(0,.9)(12,.9)};\end{axis}\end{tikzpicture}\par\smallskip{\footnotesize (b) Mapped city with ray tracing: 4,511 common truths}\end{minipage}
\caption{Accuracy on identical target IDs detected by both methods. Full-band site~1 is the strongest fixed site by detection in each population. Its transmitted PRS budget equals aggregate matched cooperation, which uses three receive apertures. Position P90 is 8.607/2.522 m in (a) and 8.536/3.497 m in (b). These gains survive stronger baselines and population control; they are not attributed to staging alone.}
\label{fig:common_accuracy}\end{figure*}
\subsection{Geometry, Bias, and Correlation Controls}
Figure~\ref{fig:controlled} first tests the theory independently of target
association. A target at $(0,0,120)$ m has velocity $(25,-15,8)$ m/s.
Simultaneous monostatic nodes lie on a 500 m horizontal circle at 25 m height.
Varying azimuth span and node count from two to twelve verifies
\eqref{eq:rank}: two views have deficient 3D velocity rank, while narrow-angle
rank-six geometry can have poor precision. At a $15^\circ$ span,
the RMS unconstrained velocity CRB remains 78.9 m/s even with twelve nodes;
three full-circle nodes give 0.54 m/s. Two views at epochs separated by 50 ms
have rank six through displacement information.

The risk fixture linearizes two additional monostatic views at
$(500,0,25)$ and $(0,500,25)$ m using \eqref{eq:jacobian} and monostatic scaling.
It uses $\mathbf P=\operatorname{diag}(4,2,0.5,0.16,0.49,0.25)$,
$\mathbf W=\operatorname{diag}(1/25,1/25,1/25,4,4,4)$, and
$\boldsymbol\epsilon=0.25\mathbf H\mathbf e_S+\boldsymbol\epsilon_0$.
Independent $\boldsymbol\epsilon_0$ has the range/angle/radial standard
deviations stated in Fig.~\ref{fig:controlled}(a). At design radius two,
worst normalized MSE is 8.531 for nominal addition, 3.860 for locking,
3.227 for optimized scalar shrinkage, and 2.997 for directional refinement.
Directional refinement improves worst risk by 22.3\% over locking and 7.1\%
over scalar shrinkage; 200,000 noise realizations verify the curves. Varying shared-error loading
to $\boldsymbol\epsilon=\zeta\mathbf H\mathbf e_S+\boldsymbol\epsilon_0$
demonstrates the role of $\mathbf C$: at $\zeta=1$, innovations contain no
baseline-error information, whereas ignoring correlation still admits
an erroneous update.

\subsection{Local Rejection and Cooperation in the 3GPP Environment}
Figure~\ref{fig:ablation} evaluates local evidence. At full-band
TRP0, the seven checks reduce mean unmatched outputs from 30.671 to 0.303
per occupied drop and FA2 from 84.82\% to 5.16\%; detection decreases
97.52\% to 94.41\%. Temporal, spatial, TX-signature, and delayed-companion
checks provide the largest interaction-adjusted reductions. Split-speed and
split-power checks have little marginal contribution at this setting.

Table~\ref{tab:primary_results} and Fig.~\ref{fig:selection_tradeoff}
separate selection, corroboration, and recovery. Observable selection detects
5,014 truths with 32 unmatched outputs. SC detects 4,926 with two unmatched
outputs. SAR adds 30 matches and three unmatched outputs; PC has the same
4,956 matches but 77 unmatched outputs. Strong anchoring suppresses most
all-weak false groups, while corroboration incurs a coverage cost relative
to selection. Full-band fixed sites 0/1/2 detect 94.41/95.25/95.23\%.
Against the strongest fixed site, matched SAR increases detection by 1.543
percentage points, paired interval [0.918,2.168]. On 4,771 shared truths,
position P90 falls from 8.607 to 2.522 m. Figure~\ref{fig:common_accuracy}(a)
shows the corresponding cumulative distribution functions (CDFs).
Multiple receive apertures and geometric estimation improve precision
despite splitting transmitted resources. Independent 128-drop occupied
validation gives SC 615/0 and AC/SAR 619/0 matched/unmatched outputs.

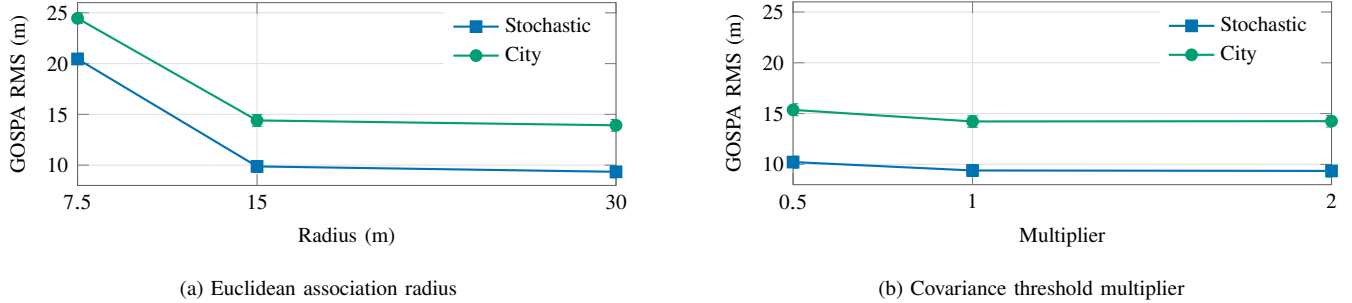
\begin{figure*}[t]\centering
\begin{minipage}[b]{.48\textwidth}\centering\begin{tikzpicture}\begin{axis}[paperaxis,width=\linewidth,height=4.0cm,xlabel={Radius (m)},xmin=7.5,xmax=30,xtick={7.5,15,30},ylabel={GOSPA RMS (m)},ymin=8,ymax=26,legend style={font=\footnotesize,at={(.98,.98)},anchor=north east}]
\addplot[paperblue,mark=square*,thick,error bars/.cd,y dir=both,y explicit] coordinates {(7.5,20.452685) +=(0,0.56346685) -=(0,0.54594965) (15,9.8733245) +=(0,0.55010191) -=(0,0.58162049) (30,9.3452898) +=(0,0.56482642) -=(0,0.59869571)};\addlegendentry{Stochastic}
\addplot[papergreen,mark=*,thick,error bars/.cd,y dir=both,y explicit] coordinates {(7.5,24.458531) +=(0,0.51714246) -=(0,0.49869616) (15,14.403468) +=(0,0.57194034) -=(0,0.58106044) (30,13.9231) +=(0,0.5698842) -=(0,0.58164541)};\addlegendentry{City}
\end{axis}\end{tikzpicture}\par\smallskip{\footnotesize (a) Euclidean association radius}\end{minipage}
\hfill
\begin{minipage}[b]{.48\textwidth}\centering\begin{tikzpicture}\begin{axis}[paperaxis,width=\linewidth,height=4.0cm,xlabel={Multiplier},xmin=.5,xmax=2,xtick={.5,1,2},ylabel={GOSPA RMS (m)},ymin=8,ymax=26,legend style={font=\footnotesize,at={(.98,.98)},anchor=north east}]
\addplot[paperblue,mark=square*,thick,error bars/.cd,y dir=both,y explicit] coordinates {(0.5,10.224828) +=(0,0.57533237) -=(0,0.60408911) (1,9.3970816) +=(0,0.55182139) -=(0,0.59504613) (2,9.351935) +=(0,0.57169895) -=(0,0.59254988)};\addlegendentry{Stochastic}
\addplot[papergreen,mark=*,thick,error bars/.cd,y dir=both,y explicit] coordinates {(0.5,15.362703) +=(0,0.59907428) -=(0,0.56822244) (1,14.222909) +=(0,0.55741557) -=(0,0.57516789) (2,14.252819) +=(0,0.55651326) -=(0,0.57551852)};\addlegendentry{City}
\end{axis}\end{tikzpicture}\par\smallskip{\footnotesize (b) Covariance threshold multiplier}\end{minipage}
\caption{Association sensitivity with identical local measurements and SAR processing, 1,024 occupied/1,024 empty drops per environment. Error bars are 95\% whole-drop intervals. (a) The 7.5/15/30 m extension is exploratory; all radii are shown. The 30 m gate is a strong simple comparator. (b) Scaling the nominal covariance gate reveals a broad plateau and the loss caused by over-tight screening. Neither curve optimizes the joint receiver threshold.}
\label{fig:gates}\end{figure*}

\begin{table*}[t]\centering
\caption{Factor controls on identical matched-budget reports with SAR. Own-matched P90 and set error are different views. Euclidean radii are shown in Fig.~\ref{fig:gates}; 30 m is exploratory on the main populations and fixed on independent validation.}
\label{tab:controls}\setlength{\tabcolsep}{4pt}
\begin{tabular}{lrrrrrrrr}\toprule
&\multicolumn{4}{c}{3GPP stochastic environment}&\multicolumn{4}{c}{Mapped city with ray tracing}\\
Control&Matches/false&Pos. P90&Vel. P90&GOSPA&Matches/false&Pos. P90&Vel. P90&GOSPA\\
& &m&m/s&m& &m&m/s&m\\\midrule
Covariance + Huber & 4,956/5 & 2.630 & 0.301 & 9.397 & 4,734/22 & 3.531 & 0.320 & 14.223\\
Covariance + LS & 4,956/5 & 2.604 & 0.300 & 9.390 & 4,734/22 & 3.511 & 0.317 & 14.210\\
Euclidean 15 m + Huber & 4,930/1 & 2.713 & 0.335 & 9.873 & 4,699/2 & 3.558 & 0.339 & 14.403\\
Euclidean 30 m + Huber & 4,955/2 & 2.619 & 0.300 & 9.345 & 4,736/5 & 3.524 & 0.320 & 13.923\\
\bottomrule\end{tabular}\end{table*}

\begin{figure*}[t]\centering
\begin{minipage}[b]{.48\textwidth}\centering\begin{tikzpicture}\begin{axis}[paperaxis,width=\linewidth,height=4.0cm,width=.95\linewidth,yticklabel style={/pgf/number format/fixed,/pgf/number format/precision=2},xlabel={Design bias radius $\rho$},ylabel={Position MSE change (m$^2$)},xmin=0,xmax=8,xtick={0,1,2,4,8},legend style={font=\footnotesize,at={(.98,.02)},anchor=south east}]
\addplot[paperblue,mark=square*,thick,error bars/.cd,y dir=both,y explicit] coordinates {(0.0,-0.047700714) +=(0,0.031401768) -=(0,0.034839606) (0.5,-0.048349906) +=(0,0.028990213) -=(0,0.033821262) (1.0,-0.039814027) +=(0,0.022673339) -=(0,0.02615529) (2.0,-0.021920943) +=(0,0.01194115) -=(0,0.014188418) (4.0,-0.0091536826) +=(0,0.0046161846) -=(0,0.0054167832) (8.0,-0.0030447294) +=(0,0.0013895001) -=(0,0.0016385083)};\addlegendentry{Stochastic}
\addplot[papergreen,mark=*,thick,error bars/.cd,y dir=both,y explicit] coordinates {(0.0,-0.0083216562) +=(0,0.059000712) -=(0,0.057413579) (0.5,-0.043570779) +=(0,0.047616833) -=(0,0.047647326) (1.0,-0.061032128) +=(0,0.035751658) -=(0,0.035784979) (2.0,-0.051669884) +=(0,0.023689712) -=(0,0.023034927) (4.0,-0.026768389) +=(0,0.012636756) -=(0,0.012190829) (8.0,-0.011282336) +=(0,0.0051906944) -=(0,0.0053781378)};\addlegendentry{City}
\addplot[paperblack,dotted,forget plot] coordinates {(0,0)(8,0)};
\end{axis}\end{tikzpicture}\par\smallskip{\footnotesize (a) Position}\end{minipage}
\hfill
\begin{minipage}[b]{.48\textwidth}\centering\begin{tikzpicture}\begin{axis}[paperaxis,width=\linewidth,height=4.0cm,width=.95\linewidth,yticklabel style={/pgf/number format/fixed,/pgf/number format/precision=2},xlabel={Design bias radius $\rho$},ylabel={Velocity MSE change ((m/s)$^2$)},xmin=0,xmax=8,xtick={0,1,2,4,8},legend style={font=\footnotesize,at={(.98,.02)},anchor=south east}]
\addplot[paperblue,mark=square*,thick,error bars/.cd,y dir=both,y explicit] coordinates {(0.0,-0.016891625) +=(0,0.013779046) -=(0,0.018400018) (0.5,-0.016307555) +=(0,0.013243739) -=(0,0.01756862) (1.0,-0.012832916) +=(0,0.010359506) -=(0,0.013745605) (2.0,-0.006229223) +=(0,0.0049748249) -=(0,0.006622317) (4.0,-0.0020010539) +=(0,0.0015805998) -=(0,0.0021016167) (8.0,-0.00053913512) +=(0,0.00042264846) -=(0,0.00056182559)};\addlegendentry{Stochastic}
\addplot[papergreen,mark=*,thick,error bars/.cd,y dir=both,y explicit] coordinates {(0.0,-0.053882104) +=(0,0.020105658) -=(0,0.024281356) (0.5,-0.051784738) +=(0,0.019432368) -=(0,0.023444092) (1.0,-0.040525382) +=(0,0.015207387) -=(0,0.018350378) (2.0,-0.019583572) +=(0,0.0073060394) -=(0,0.0088033894) (4.0,-0.0063104989) +=(0,0.0023202713) -=(0,0.0028137202) (8.0,-0.0017319079) +=(0,0.00062667161) -=(0,0.00075739157)};\addlegendentry{City}
\addplot[paperblack,dotted,forget plot] coordinates {(0,0)(8,0)};
\end{axis}\end{tikzpicture}\par\smallskip{\footnotesize (b) Horizontal vector velocity}\end{minipage}
\caption{Protected refinement on the actual matched-budget Euclidean-30 m report pools. Changes are relative to SAR locking, on identical matched truth IDs; negative favors refinement. Only 63 stochastic and 268 city strict states receive additional weak views. Bars are paired 95\% whole-drop intervals. Smaller $\rho$ permits stronger refits; larger $\rho$ approaches locking. }
\label{fig:refinement}\end{figure*}

\subsection{Target Returns and Sensing in the Mapped City}
The city has native target paths at all three sites for 4,936 of 5,120 UAVs,
but locally matched strong reports at all three for only 2,856. Traced
interactions need not yield useful measurements. Local
rejection reduces full-band TRP0 false count 8.742 to 1.904 per drop,
with detection decreasing 84.20\% to 80.59\%. Spatial, temporal, and TX
consistency have the largest false-count contributions, differing from the
stochastic ranking in Fig.~\ref{fig:ablation}.

Matched selection gives 88.75\% detection and 11.06\% FA2; SAR gives
92.46\% and 0.396\%. Relative to SC, weak recovery adds 169 matches and
15 unmatched outputs: 8.15\% of the 184 extra confirmations are unmatched.
Full-band fixed-site detection
is 80.59/91.84/90.55\%. SAR's 0.625-point advantage over site~1 has interval
[$-0.176$,1.445], so positive detection improvement is unresolved. Common-truth
position P90 nevertheless falls from 8.536 to 3.497 m on 4,511 targets
(Fig.~\ref{fig:common_accuracy}(b)).
Cleaner, more accurate confirmation is more consistent than a coverage gain.

In a 64-drop/320-target path control, removing stationary background
leaves every scored output unchanged. Retaining only direct target paths
changes SC 281/0 to 294/0 and SAR 290/1 to 295/0. Target-interacting
reflections account for the remaining ambiguity in that subset, consistent
with the mechanism in Corollary~\ref{cor:image_floor}. All coded fusion
methods have zero observed empty alarms.

\subsection{Association Sensitivity and Estimator Comparisons}
Figure~\ref{fig:gates} and Table~\ref{tab:controls} isolate processing choices.
LS preserves Huber's counts and matched truth IDs
and slightly improves GOSPA; Huber adds no demonstrated benefit after hard
joint-residual gates. Halving the covariance pair threshold raises GOSPA by 0.828/1.140 m in
the stochastic/city cases; both paired intervals exclude zero. Doubling it
changes only two net matches per environment. Over-tight screening discards
useful views before fitting.

The initially specified Euclidean 15 m gate removes ghosts but loses targets. An exploratory
7.5/15/30 m sweep shows that 30 m gives 4,955/2 and 4,736/5 matched/unmatched
outputs, competitive detection with fewer false outputs than nominal
covariance gates. City GOSPA improves by 0.300 m, interval [0.174,0.444].
Transferring the fixed 30 m gate unchanged to independent 128-drop stochastic
validation preserves the nominal groups and metrics. Covariance weighting
without path bias therefore does not establish superior association.

\subsection{Bias-Aware Refinement and Common-Target Accuracy}
Refinement uses 5 m/0.5 m/s normalized loss, $\mathbf C=0$, linearized
weight proxies, and radii $0,0.5,1,2,4,8,\infty$.
AC weak augmentations preserve SC membership and exclude weak IDs used by
recovery. Eligible pools are accessible; the recovery traffic savings
below exclude additional refinement acquisition.

Figure~\ref{fig:refinement} uses the strong Euclidean-30 m baseline, with
63 stochastic and 268 city refinable states. At radius two, city common-truth
position and velocity MSE changes are $-0.0517$ m$^2$ [${-0.0747},{-0.0280}$]
and $-0.0196$ (m/s)$^2$ [${-0.0284},{-0.0123}$], relative to locking.
Both intervals exclude zero; matched and false counts remain unchanged.
The covariance-gated baseline likewise improves by $-0.0486$ m$^2$ and
$-0.0202$ (m/s)$^2$. Nominal addition gives the strongest velocity gain but
no resolved city position gain. Finite radius reduces position exposure
while retaining velocity benefit; infinite radius reproduces SAR exactly
on all 4,096 original drops.

Scalar shrinkage gives slightly lower city position MSE at radius two.
Corollary~\ref{cor:cap} concerns worst weighted risk under specified
bias/moments, as the fixture verifies; waveform reports show the practical
tradeoff.

\subsection{Reporting Payload and Request Overhead}
Locked covariance-gated SAR uploads 13,809/13,434 stochastic/city strong
records. Weak uploads fall from 1,367/1,558 to 35/213, requiring
644/2,966 unicast cue copies. Equation~\eqref{eq:traffic} gives cue/report
size breakpoints 2.068/0.4535. Equal-size records reduce stochastic traffic
4.53\% but increase city traffic 10.81\%; quarter-size city cues save 4.03\%,
excluding framing and acknowledgments. Compact requests determine net
savings; this payload differs from the 5.3375\% radiated PRS overhead.

\section{Conclusion}
Cooperative ISAC depends on diverse views and disciplined use of their evidence. This paper connects report distinguishability, shared-image ambiguity, correlated-innovation risk, and acquisition cost in a general fusion framework. Strong reports establish associations, additional views refine confirmed states under a declared bias model, and unused anchors recover further targets. UAV evaluations under a 3GPP stochastic channel model and ray tracing in a San Francisco urban scene demonstrate the benefits of cooperative fusion. In the latter environment, fusion achieves 92.46\% detection with a 0.40\% false-output fraction and reduces common-target 90th-percentile position error from 8.54 to 3.50 m. The results support separating association validity from refinement utility when designing sensing reports and network processing. Future work will incorporate map-informed bias sets and bistatic reference errors into adaptive fusion and acquisition.

\end{document}